\documentclass[10pt,journal,compsoc]{IEEEtran}

\ifCLASSOPTIONcompsoc
  \usepackage[nocompress]{cite}
\else
  \usepackage{cite}
\fi

\usepackage{url}
\usepackage{caption}
\usepackage{amsmath,amssymb,amsfonts}
\usepackage{subcaption}
\usepackage{algorithmic}
\usepackage{graphicx}
\usepackage{textcomp}
\usepackage{xcolor}
\usepackage[vlined, ruled, boxed,linesnumbered]{algorithm2e}
\usepackage{amsthm}
\SetKw{Break}{break}
\let\oldnl\nl%
\newcommand\nonl{%
	\renewcommand{\nl}{\let\nl\oldnl}}%
\newtheorem{thm}{Theorem}

\ifCLASSINFOpdf
\else
\fi

\begin{document}
\title{Efficient Large-Scale Multiple Migration Planning and Scheduling in \\SDN-enabled Edge Computing}

\author{TianZhang~He,
	Adel N. Toosi,~\IEEEmembership{Member,~IEEE,}
	Rajkumar~Buyya,~\IEEEmembership{Fellow,~IEEE}%
	\IEEEcompsocitemizethanks{\IEEEcompsocthanksitem T.Z. He and R. Buyya are with the CLOUDS Lab, School of Computing and Information Systems, University of Melbourne, Australia.
		(E-mail: tianzhangh@student.unimelb.edu.au; rbuyya@unimelb.edu.au)
		\IEEEcompsocthanksitem A. N. Toosi is with the Department of Software Systems and Cybersecurity, Faculty of Information Technology, Monash University, Australia. (E-mail: adel.n.toosi@monash.edu)}%
}

\IEEEtitleabstractindextext{%
\begin{abstract}
The containerized services allocated in the mobile edge clouds bring up the opportunity for large-scale and real-time applications to have low latency responses. Meanwhile, live container migration is introduced to support dynamic resource management and users' mobility.
However, with the expansion of network topology scale and increasing migration requests, the current multiple migration planning and scheduling algorithms of cloud data centers can not suit large-scale scenarios in edge computing. The user mobility-induced live migrations in edge computing require near real-time level scheduling. 
Therefore, in this paper, through the Software-Defined Networking (SDN)  controller, the resource competitions among live migrations are modeled as a dynamic resource dependency graph. 
We propose an iterative Maximal Independent Set (MIS)-based multiple migration planning and scheduling algorithm.
Using real-world mobility traces of taxis and telecom base station coordinates, the evaluation results indicate that our solution can efficiently schedule multiple live container migrations in large-scale edge computing environments. It improves the processing time by 3000 times compared with the state-of-the-art migration planning algorithm in clouds while providing guaranteed migration performance for time-critical migrations.
\end{abstract}

}

\maketitle

\IEEEdisplaynontitleabstractindextext

\IEEEpeerreviewmaketitle

\IEEEraisesectionheading{\section{Introduction}\label{sec:introduction}}
The introduction of edge computing \cite{satyanarayanan2017emergence} brings opportunities to improve the performance of the emerging user-oriented applications by pushing computation and intelligence to end-users, including Vehicle to Cloud (V2C), Vehicle to Vehicle (V2V), Virtual Reality (VR), Augmented Reality (AR), Artificial Intelligent (AI), or Internet of Things (IoT) applications and so forth. Driven by container virtualization, microservices are more suitable for dynamic deployment on edge computing \cite{machen2017live,doan2019containers} due to smaller memory footprint and faster startup. By allocating the containerized services in the Edge Data Centers (EDCs) or Mobile Edge Clouds (MECs) \cite{hu2015mobile}, strict end-to-end (E2E) communication delays between end-users and services can be guaranteed.

From the centralized cloud computing framework to decentralized edge computing, surveys \cite{wang2018survey,rejiba2019survey} investigated the challenges faced by the infrastructure and service providers regarding dynamic resource management and user mobility. 
By providing non-application-specific compute and memory state management, live migration is the solution to these challenges. 
Live migration of VM \cite{clark2005live} and container \cite{mirkin2008containers} through the open-source Checkpoint/Restore in Userspace (CRIU) software \cite{criu}, which had kernel support since Linux 3.11, aims to provide little or no disruption to the running service during migrating in the edge computing.  
It iteratively copies unfinished computation tasks with intermediate computation states in the memory from source to destination until the memory difference between two synchronizing instances is small enough for the stop-and-copy phase. In addition, for the container image, if the image does not exist in the destination, it can be transferred from the previous EDC or remote clouds or accessed through shared storage. 
Thus, live migration performance highly relies on the available bandwidth of the network routing connecting source to destination.

Industrial infrastructure and service providers, such as IBM, RedHat, Google, etc, have been integrating live container migration into their productions~\cite{redhat-criu,borg-criu}. Google has adopted live VM and container migration with CRIU into Borg cluster manager \cite{verma2015large,ruprecht2018vm,borg-criu} for reasons, such as, higher priority task preemption, software updates, such as kernel and firmware, or reallocating for availability or performance. It manages all compute tasks and runs on numerous container clusters each with up to tens of thousands of machines. A lower bound of 1,000,000 migrations monthly in the production fleet have been performed with 50ms median blackout~\cite{ruprecht2018vm}. Live container migration provides technical simplicity without handling state management and application-specific evictions. However, it is also identified that writing and reading to remote storage through network dominates the checkpoint/restore process and the scheduling delay is the large source of delay regarding the performance of multiple live migrations.

\begin{figure}[th]
	\centering
	\includegraphics[width=0.8\linewidth]{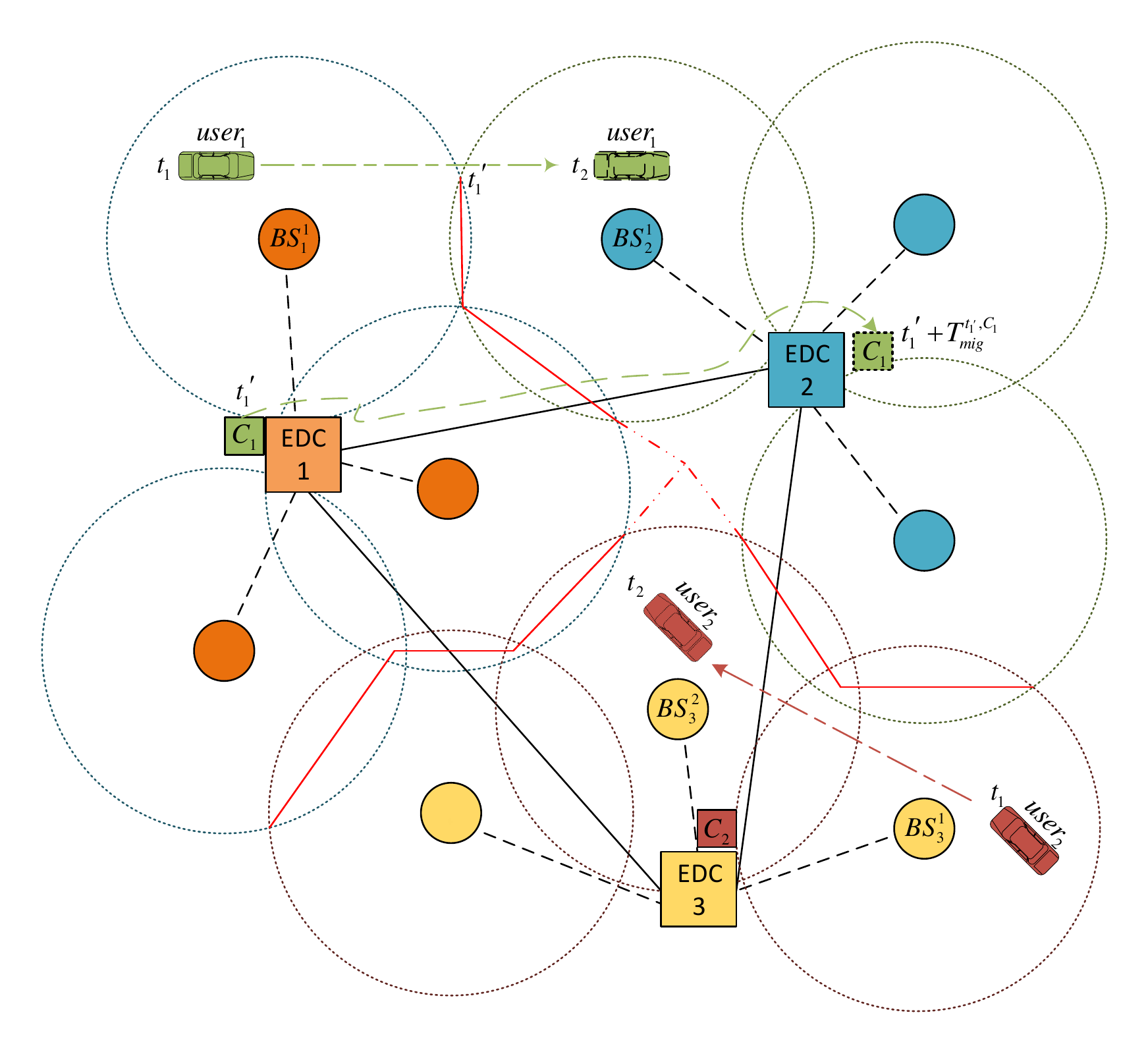}
	\caption{User-mobility induced live container migration in edge computing environments}
	\label{fig: introduction}
\end{figure}

Recently, some works have focused on user or service mobility in mobile edge computing through live container migration~\cite{wang2015dynamic,ma2017efficient, ma2018efficient,wang2018survey,rejiba2019survey,wang2019dynamic,wang2019delay}.
With the limited coverage range of each EDC, when a user moves from one base station to another, the network latency could deteriorate after the network handover. To guarantee the Quality of Service (QoS), the service may need to be migrated from the previous EDC to the proximal one through live container migration.
Figure \ref{fig: introduction} illustrates an example scenario where an autonomous vehicle sends workload to the corresponding stateful service in the EDC for real-time object detection. There is a total of 9 base stations assigned to 3 different EDCs. Two autonomous vehicles move to a new position from time $t_1$ to $t_2$. For vehicle user1, when leaving the range of base station $BS_1^1$ and entering the range of base station $BS_2^1$ at time ${t_1}^\prime$, there is a live container migration from $EDC1$ to $EDC2$ induced by the user1's movement. Meanwhile, as user2 crosses the range boundary of the base station $BS_3^1$ to $BS_3^2$, since the two base stations both belong to the same edge data center $EDC3$, the E2E delays of the service can be guaranteed. Therefore, there is no live migration induced by user2's movement.

In cloud computing environments, dynamic resource management policies triggers live migration requests periodically to optimize the resource usage or to maintain QoS of applications~\cite{beloglazov2012optimal}. However, in the edge environment, service migration/mobility is highly relative to user mobility~\cite{wang2015dynamic,ma2017efficient,wang2019dynamic,wang2019delay,rejiba2019survey}. 
Migration requests of containers or VMs from services and users may share and compete for the computing and network resources, such as migration source, destination, and network routings~\cite{he2021sla}. 
It brings more challenges for the multiple live migration planning and scheduling. 
However, most research on live service migration in edge and cloud computing neglects the actual live migration cost regarding the iterative dirty memory transmission \cite{clark2005live} and resource competition among migrations in both computing and network resources. 
As a result, performing multiple live migrations in arbitrary order can lead to service degradation \cite{bari2014cqncr, wang2017virtualold}. 

Few works focus on multiple VM migration planning and scheduling in cloud data centers \cite{bari2014cqncr, wang2017virtualold}.
The framework of migration scheduling periodically triggered by resource management policies with a long time interval is not suitable for stochastic scenarios of mobility-induced migration in edge computing. Furthermore, the network scale, the numbers of end-users and live container migration requests increase ten thousand times in edge environments. Without proper modeling, the problem complexity will increase dramatically as the number of migration requests and network scale increase. As a result, the complexity and processing time of current algorithms do not meet the real-time requirement of the live migration at scale in edge environments.

Moreover, to manage a highly distributed and dynamic network environment, Software-Defined Networking (SDN) \cite{son2017taxonomy} is introduced in edge computing, which allows dynamic configuration and operations of computer networks through centralized software controllers. 
Facilitated by OpenFlow protocol \cite{mckeown2008openflow} and Open vSwitch (OVS) \cite{pfaff2015design}, network manager based on SDN controllers 
can perform network slicing or create a separate network \cite{ordonez2017network} to minimize the influence of migration flows on other edge services.
As a result, the migration planning and scheduling algorithm can fine-grained control the network resources for the migration competition, including the network routing and available bandwidth.

Therefore, in this paper, we propose efficient large-scale live container migration planning and scheduling algorithms focusing on mobility-induced migrations in edge computing environments. It can still apply to multiple migration scheduling for the general dynamic resource management at scale.
The \textbf{contributions} of this paper are summarized as follows:
\begin{itemize}
	\item We introduce the resource dependency graph of the source-destination pair for resource competition among migration requests to reduce the problem complexity.
	\item We model the problem as finding the Maximum Independent Set of the Dependency Graph iteratively.
	\item We propose iterative-Maximal Independent Set (MIS)-based algorithms for efficient large-scale migration scheduling and prove the corresponding Thermos.
	\item We implement an event-driven simulator to evaluate the user mobility and live container migrations. The experiments are conducted with real-world dataset and traces.
\end{itemize}

The rest paper is organized as follows. We review the related work in Section \ref{section: related}  and define the system architecture in Section \ref{section: system}. In Section \ref{section: problem-model}, we analyze and model the problem of multiple container migration scheduling. Then we propose two main methods of large-scale migration scheduling in Section \ref{section: multi-scheduling}. Section \ref{section: evaluation} shows the performance analysis of proposed algorithms and Section \ref{section: experiment} shows the experimental design and evaluation with real-world dataset. Finally, we conclude the paper in Section \ref{section: conclusion}.

\section{Related Work} \label{section: related}
\begin{table*}[t]
	\caption{Comparisons of multiple migration planning and scheduling works}
	\label{tb: relate-works}
	\centering
	\resizebox{\linewidth}{!}{%
		\begin{tabular}{|l|ll|llll|ll|}
			\hline
			research		&	real-time planning	&	large-scale		&	service correlations 	& user-mobility 	& deadline  & SDN-enabled		&	Cloud DC		&  Edge computing		\\
			\hline
			CQNCR \cite{bari2014cqncr}		&	$-$	&	$-$		&	$\checkmark$		&	$-$		&	$-$						&		$-$			& $\checkmark$	&	$-$ \\
			FPTAS \cite{wang2017virtualold} 	&	$-$	&	$-$		&	$-$					&	$-$		&	$-$						&		$\checkmark$& $\checkmark$	&	$-$ \\
			Our work						&	$\checkmark$	&	$\checkmark$&	$\checkmark$			&	$\checkmark$		&	$\checkmark$&		$\checkmark$& $\checkmark$	&	$\checkmark$ \\
			\hline
		\end{tabular}
	}
\end{table*}

The live VM migration realization \cite{medina2014survey} and its application in cloud data centers, such as dynamic resource management \cite{mishra2012dynamic,xu2017survey}, have been matured last few years. The research on live container migration in edge computing is an active field \cite{wang2018survey,rejiba2019survey}.
Clark et al. \cite{clark2005live} proposed the live VM migrations and discussed the details of pre-copy or iterative live migration. He et al. \cite{he2019performance} evaluated the performance of live VM migrations and its overheads on the migration services in SDN-enabled cloud data centers.
On the other hand, the research on live container migration is trending and becomes more mature in recent years. 
Mirkin et al. \cite{mirkin2008containers} represented the checkpointing and restart features for the OpenVZ container. The checkpointing function is also used for live migration.
Checkpoint/Restore In Userspace (CRIU) \cite{criu} is a Linux software to migrate container's in-memory state in userspace. It is currently integrated with LXC, Docker (runC), and OpenVZ to achieve the live container migration. 
Nadgowda et al. also proposed \cite{nadgowda2017voyager} a CRIU-based memory migration together with
the data federation capabilities of union mounts to minimize migration downtime. 
Similarly, Ma et al. \cite{ma2017efficient, ma2018efficient} utilized the layered storage feature based on AUFS storage drive and implemented a prototype system to improve the performance of docker container migration.
Furthermore, several works studied the performance difference between container and VM live migration \cite{machen2017live, doan2019containers}. Results show that the live container migration is much faster than the live VM migration due to its much smaller memory footprint and fast startup features.

More research recently focuses on dynamic resource scheduling in fog and edge computing environments based on the live container migration (details in survey \cite{wang2018survey, rejiba2019survey}).
In \cite{wang2015dynamic, wang2019dynamic}, the authors modeled the sequential decision making
problem of generating service migration requests using the distance-based Markov Decision Process (MDP) framework. By reducing the state space, they proposed a distance-based MDP to get the approximated results.
The research \cite{wang2019delay} also investigated the same problem by using the MDP framework. The authors proposed a reinforcement learning-based online microservice coordination algorithm to learn the optimal strategy for live migration requests to maintain the QoS in end-to-end delay.

Few studies focus on the optimization of the multiple live VM migration planning in cloud data centers \cite{bari2014cqncr, wang2017virtualold}. Bari et al. \cite{bari2014cqncr} investigated the multiple VM migration planning in one data center environment by considering the available bandwidth and the migration effects on network traffic reallocation. 
The authors proposed a heuristic migration grouping algorithm (CQNCR) by setting the group start time only based on the prediction model. Without an on-line scheduler, the estimated start time of a live migration can lead to an unacceptable migration performance and QoS degradations. Moreover, the work neglects the cost of the individual migration by only comparing the migration group cost.
Without considering the connectivity between VMs and the change of bandwidth, Wang et al. \cite{wang2017virtualold} simplified the problem by maximizing the net transmission rate rather than minimizing the total migration time and proposed a polynomial-time approximation algorithm by omitting certain variables. However, the solution (FPTAS) can create migration competing on the same network path which degrades the migration performance in both average individual migration time and the total migration time.

However, as shown in Table \ref{tb: relate-works}, current algorithms can not meet the requirement of live container migrations in edge computing.
The framework of migration scheduling \cite{bari2014cqncr} which are periodically triggered by resource management policies with a long time interval is not suitable for the mobility-induced migration scenario.
Furthermore, by modeling and calculating every resource competition of migration directly, the problem complexity \cite{bari2014cqncr, wang2017virtualold} increases along with the migration request number which is not suitable for large-scale situation. The running time of migration planning is also too large to schedule time-critical live container migrations.
The algorithms \cite{bari2014cqncr, wang2017virtualold} do not consider the deadline or urgency (priority) of migration. In addition, without an on-line scheduler, the start time of a migration schedule is only based on the estimated migration time which can lead to migration performance and QoS degradation.

\section{System Architecture} \label{section: system}
In the edge computing, there is no dedicated network for the live migration to support the user mobility compared with the traditional setups in cloud data centers \cite{tsakalozos2017live}. By integrating the Software-Defined Networking (SDN) into edge computing, the centralized SDN controller can dynamically separate network resources from the service network \cite{son2017taxonomy} to build a virtual WAN network for live migrations. The available bandwidth and network routing are dynamically allocated based on the reserved bandwidth of the service network. This solution alleviates the overheads of live migration on other services and guarantees the performance of multiple live migrations. To achieve a fine-grained live migration scheduling, the migration scheduling service is integrated with the SDN controller~\cite{he2019performance}, such as OpenDayLight (ODL), Open Network Operating System (ONOS) and Ryu, and container management and orchestration module, such as Kubernetes and Docker Swarm, to control both network and computing resources during each migration lifecycle.

\subsection{Migration Lifecycle} \label{section: mig-lifecycle}
\begin{figure}[ht]
	\centering
	\includegraphics[width=0.8\linewidth]{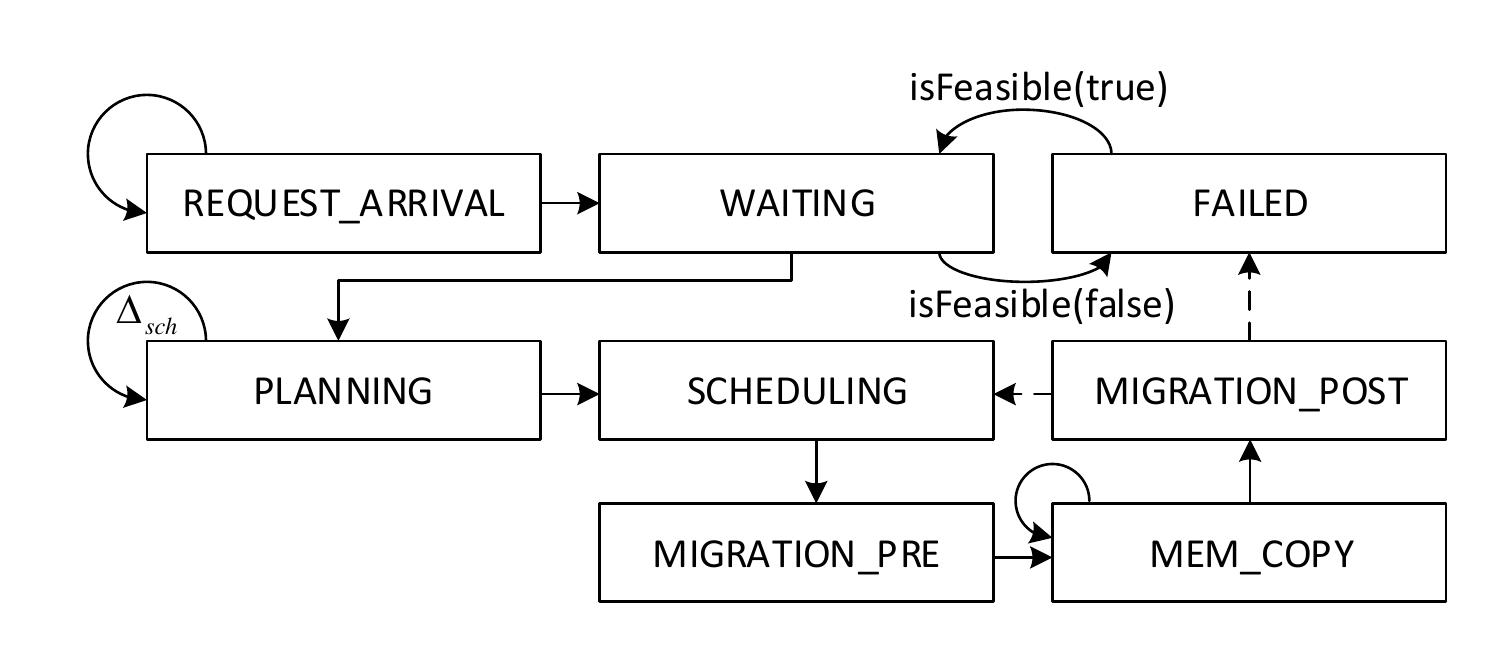}
	\caption{Lifecycle for live container migrations}
	\label{fig: life-cycle}
\end{figure}
In this section, we introduce the framework of migration scheduling in edge computing.
Compared with periodically arrived multiple live migrations in cloud data centers, the arrival of live container migration induced by user mobility is stochastic. Therefore, we design the scheduling framework for the planning and scheduling of live container migration in edge computing with stochastic environments. As shown in Fig. \ref{fig: life-cycle}, when a migration request arrives, it enters the WAITING state if it is feasible for scheduling, which means the container is not in migration. Otherwise, it will enter into the FAILED waiting migration list of the corresponding container. 
The migration planning event is triggered periodically within a short interval (such as every 1 second). It will generate the migration scheduling plan according to both waiting and running migrations. Based on the migration plan, the SDN-enabled on-line scheduler starts the migration with the allocated bandwidth and routing. Then, the live container migration will start the pre-migration phase to extract the container procedure tree. This will trace the dirty memory in the userspace of the source server and create an empty container instance in the destination for state synchronization \cite{criu}. In MEM\_COPY, the dirty memory is transferred iteratively to the synchronizing instance in the destination. In the post-migration phase, the network communication of the migrated service will be redirected to the new instance in the destination. Then, the migrated container will recover at the destination. It will also trigger the start of subsequent resource-dependent migrations in the plan and change the feasibility flag of the first migration request of the same container in the FAILED migration waiting list.

\section{Motivations and Problem Formulation} \label{section: problem-model}
In this section, we first present the performance model of single live migration. Then, we analyze the challenges faced by multiple live container migrations scheduling in edge computing: resource competition or dependency and real-time planning and scheduling. Finally, we model the problem as iteratively generating the Maximal Independent Set (MIS) based on the resource dependency graph.

\subsection{Single Migration Model} \label{section: mig-model}
The performance of single live migration $T_{mig}$ can be categorized into three parts: pre-migration
computing, memory-copy networking, and post-migration computing overheads, i.e., ${T_{mig}} = {T_{pre}} + {T_{mem}} + {T_{post}}$.
Due to the smaller footprint and fast start up of containers compared to VMs, the pre-migration and post-migration is much shorter \cite{nadgowda2017voyager}.
Based on the iterative pre-copy container migration implemented in CRIU \cite{criu}, the
migration performance in terms of memory-copy can be
represented as \cite{he2019performance}:
\begin{equation} \label{eq: mig-time}
{T_{mem}} = \frac{{\rho  \cdot {Mem}}}{L} \cdot \frac{{1 - {\sigma ^{i + 1}}}}{{1 - \sigma }}
\end{equation}
where the ratio $\sigma  = \rho  \cdot {R \mathord{\left/ {\vphantom {R L}} \right. \kern-\nulldelimiterspace} L}$, $\rho$ is the data compression rate of dirty memory, $Mem$ is amount of kernel memory the container uses, $L$ is allocated available bandwidth, $R$ is dirty page rate which is the memory difference in pagemap per second compared to the previous copy iteration, $i$ is the total migration round. 

We consider three conditions to enter the stop-and-copy phase: (1) reach the threshold of memory copy iteration; (2) the transmission time of remained memory difference is less than the downtime threshold; and (3) the allocated bandwidth is less than the dirty page rate.
The overhead of disk transmission for the container data and image can be ignored when shared network storage is available. The total iteration rounds of memory copy can be represented as:
\begin{equation}\label{eq:round-number}
i = \min \left( {\left\lceil {{{\log }_\sigma }\frac{{{V_{thd}}}}{Mem}} \right\rceil ,\Theta } \right)
\end{equation}
where $\Theta$ denotes the maximum allowed number of iteration rounds. The $\Theta = 0$ when the dirty page rate is larger than the allocated bandwidth at the start of migration. $V_{thd} = T_{dthd} \cdot L$ is the remaining dirty pages need to be transferred in the stop-and-copy phase, and $T_{dthd}$ is the configured downtime threshold.

\subsection{Resource Competition}
\begin{figure}[t]
	\centering
	\begin{subfigure}[b]{.45\linewidth}
		\centering
		\includegraphics[width=\linewidth]{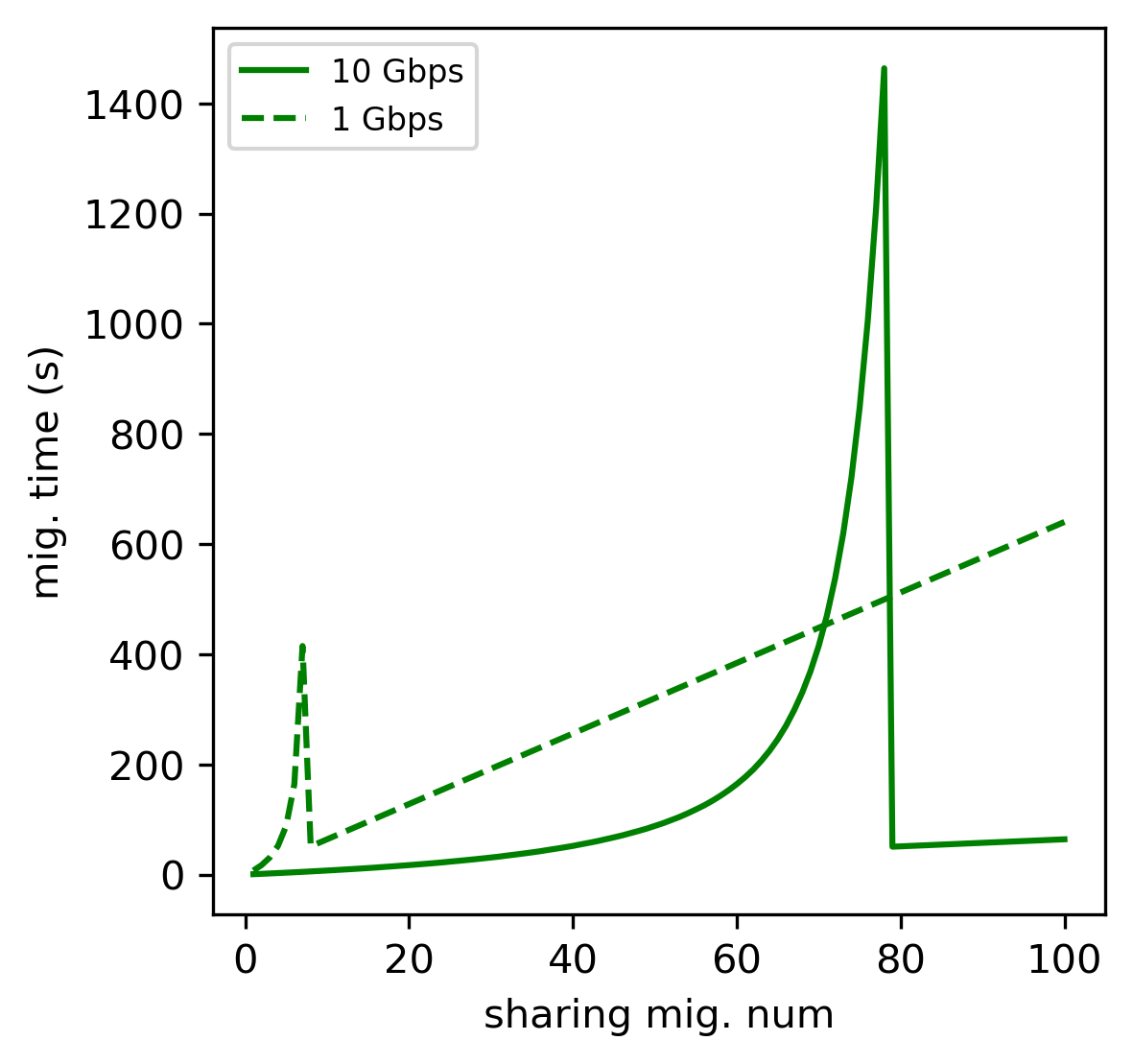}
		\caption{Average migration time}
		\label{fig: netshare1}
	\end{subfigure}%
	\hspace{0.1em}
	\begin{subfigure}[b]{.5\linewidth}
		\centering
		\includegraphics[width=\linewidth]{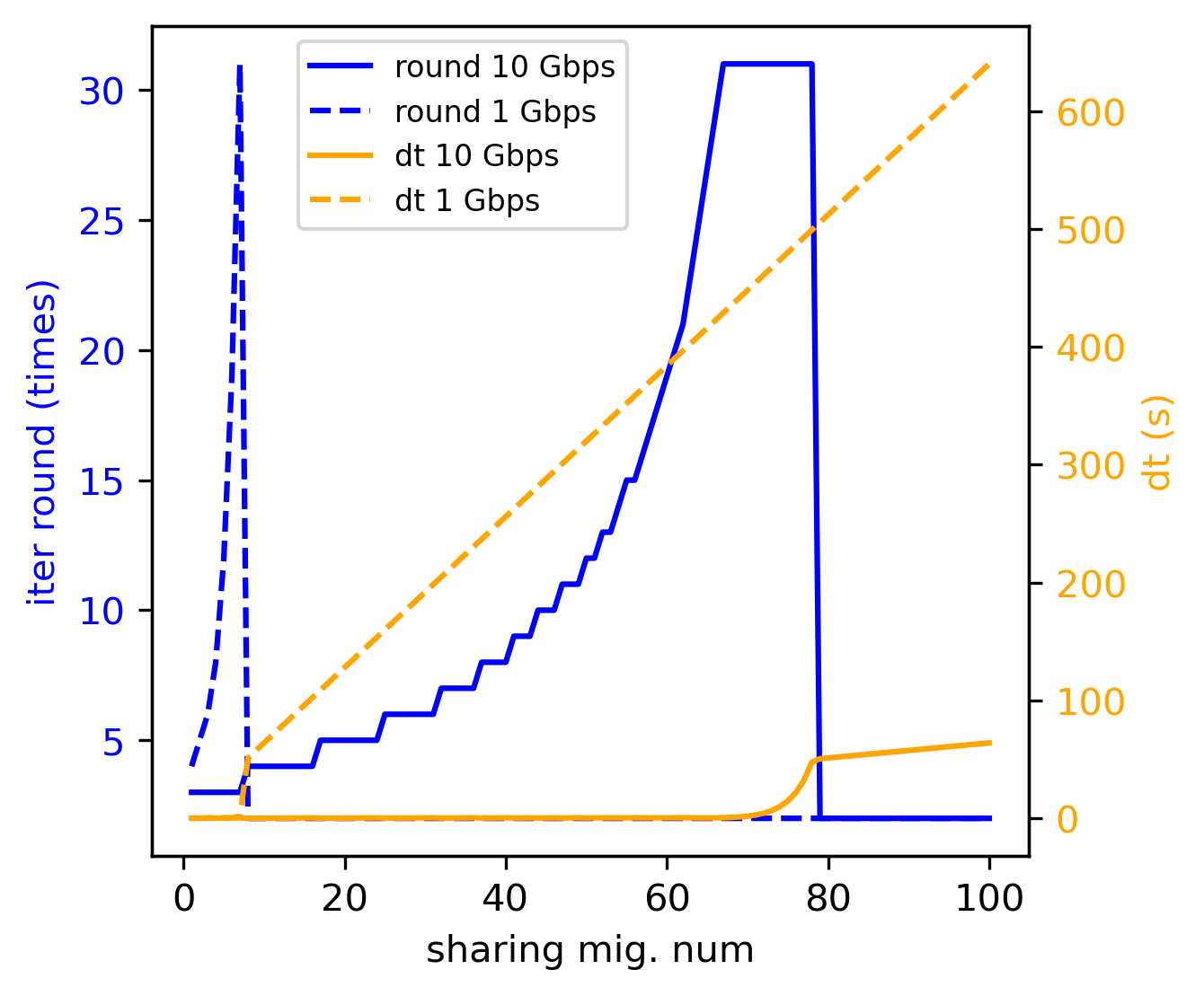}
		\caption{Iterations and downtime(dt)}
		\label{fig: netshare2}
	\end{subfigure}
	\caption{Migration performance against the number of migration sharing network bandwidth}
\end{figure}
We first explain the network sharing competition overheads in multiple migration scheduling.
Two migrations may share the same source, destination, or part of network routings. Therefore, performing multiple live migrations in arbitrary order can lead to service degradation and unacceptable migration performance \cite{bari2014cqncr, wang2017virtualold}.
A smaller bandwidth during the live migration means a longer migration time and more dirty pages need to transfer in order to limit the state difference between two instances for the last stop-and-copy phases which contributes as the downtime. Thus, the sum of the individual migration time of several live migrations is less than the total live migration time \cite{he2019performance}. For example, based on the live migration model, Fig. \ref{fig: netshare1} and \ref{fig: netshare2} show the situation when several identical migrations sharing the same network path. The container's initial memory size is 1 GB with a 20 MB/s dirty page rate. The downtime and iteration threshold is configured at 0.5 seconds and 30 times, respectively. In this example, for the sake of a clear comparison between the sum of individual migration time and the total migration time, we start all migrations at the same time. In this case, the average migration time as shown also equals the total migration time.

The average execution time of live migrations scheduled sequentially with 10 Gbps and 1 Gbps is 0.8482 and 7.5241 seconds. The average downtime is 0.0082 and 0.1048 seconds. The iteration rounds are 3 and 4, respectively.
However, the average migration time or total multiple migration time of 5 live migrations sharing 10 Gbps and 1 Gbps is 3.604 and 88.43 seconds. The average downtime is 0.2048 and 0.3689 seconds with 3 and 12 iterations, respectively.
As the number of migrations increases (Fig. \ref{fig: netshare1}), the allocated bandwidth decreases linearly. However, to achieve the required migration downtime, the average migration time will increase exponentially. At 7 and 80 migrations sharing of 1 Gbps and 10 Gbps respectively, the iteration rounds reach the threshold as 30 (Fig. \ref{fig: netshare2}). Then, with more bandwidth-sharing migrations, the dirty page rate is larger than the allocated migration bandwidth. The downtime exceeds the 0.5 seconds threshold and increases significantly from 0.78 to 64.0 seconds and from 1.85 to 640 seconds.
For the time-critical live migrations, a longer migration time will increase the possibility of migration deadline violation and QoS degradation. Therefore, it is optimal to sequentially schedule the resource-dependent migrations while concurrently schedule the independent ones. If there is a set of independent migrations and no other resource-dependent migrations are running, we can start all migration in such a concurrent scheduling group. The objective of migration scheduling is to maximize the number of migrations that can be scheduled concurrently.

\begin{figure}[t]
	\centering
	\includegraphics[width=0.6\linewidth]{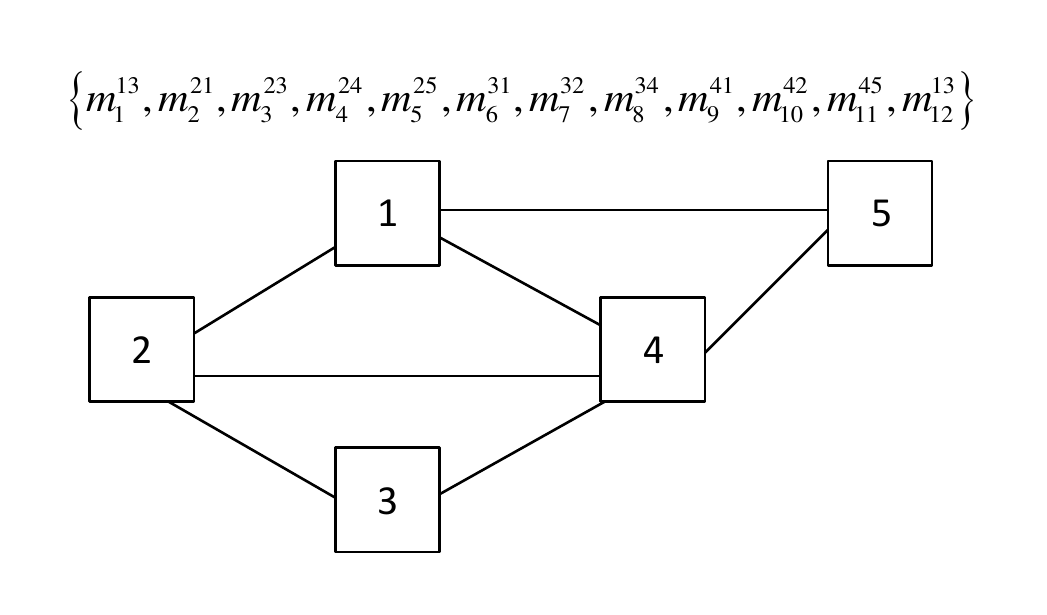}
	\caption{Example live migration requests and the network topology with 5 edge data centers}
	\label{fig: example-topo}
\end{figure}

Figure \ref{fig: example-topo} shows an illustrative example with twelve live migrations requests on the edge network topology of 5 total EDCs. Let $m^{sd}_i$ denote the migration request that migrating container $i$ from EDC $s$ to EDC $d$. For the sake of a concise example, we limit the network interfaces used by the migration traffic. In other words, migration traffics share the same interfaces when the source or the destination is the same. It can be easily extended to the set of network interfaces in source $\{s\}$ and destination servers $\{d\}$ and the corresponding network paths $\{p\}$. The network routing policy considers the shortest network path with the minimal number of migration flows.
For example, there are two network routes between EDC1 and EDC3. As there is one migration from EDC2 to EDC3, it chooses network path $\{EDC1, EDC4, EDC3\}$ in this case. 

\begin{figure}[t]
	\centering
	\includegraphics[width=0.6\linewidth]{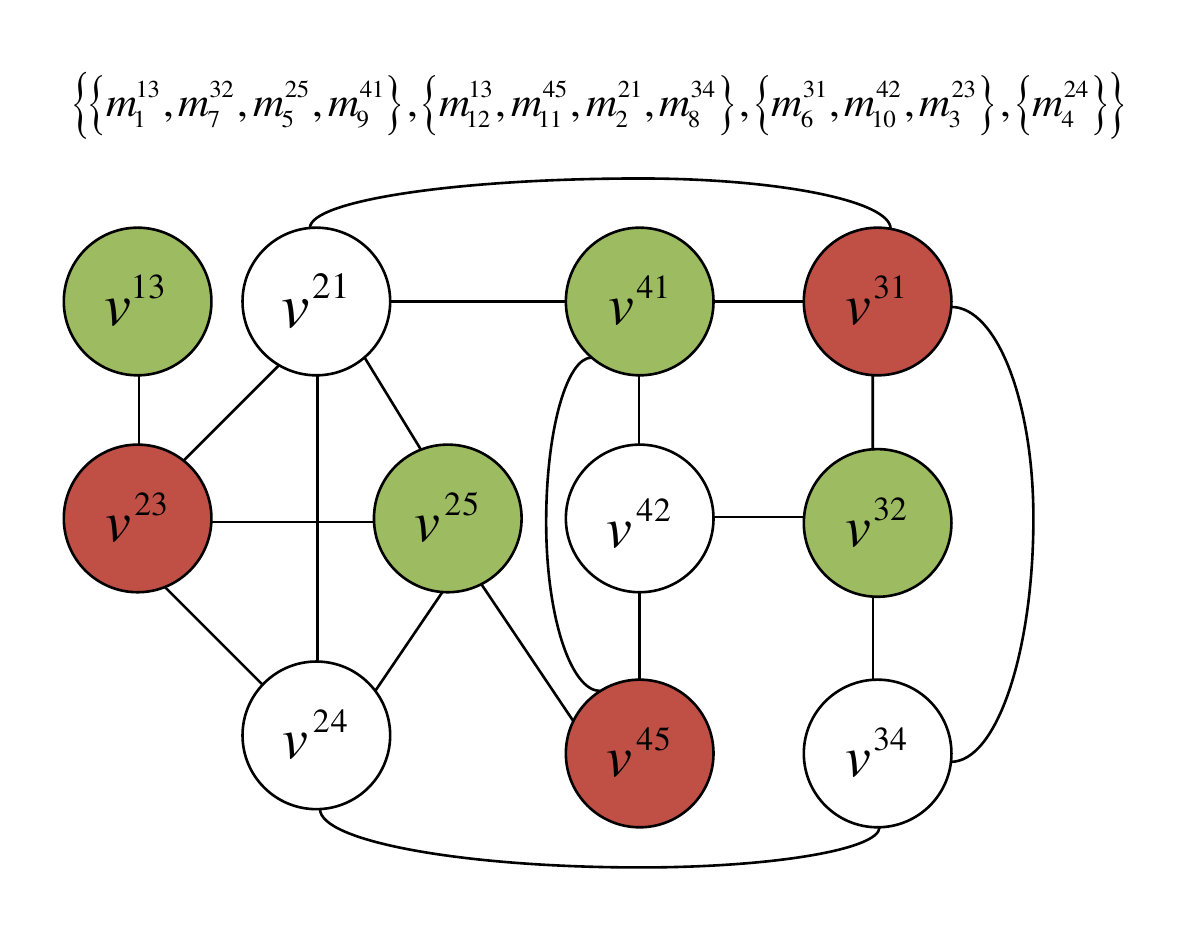}
	\caption{The resource dependency graph of example migrations, iterative maximal independent set as concurrent migration groups, and two colored possible maximal independent sets for the first iteration, and one possible concurrent migration groups}
	\label{fig: example-depgraph}
\end{figure}

Resource-independent migrations from one concurrent scheduling group can be scheduled at the same time. The planning algorithm needs to generate a scheduling plan consists of several concurrent migration groups that each group size is as large as possible. A larger concurrent scheduling group indicates that there are more migrations could be performed at the same time. As a result, the better performance of multiple migrations in total migration time and the QoS of migrating service can be guaranteed. Note that two migrations from different migration groups are not necessarily resource-dependent. As shown in Fig. \ref{fig: example-depgraph}, based on the network topology provided by the SDN controller, we create an undirected graph of resource dependency among migrations based on the source, destination, and network routing of migration requests. Each node $v^{sd}_p$ represents a list of migrations sharing the same source $s$, destination $d$, and network path $p$. In other words, migrations in one node list form a complete graph as all migrations are resource-dependent to all others in the list. For example, the migration list of node $v^{13}$ is $\{m_1^{13}, m_{12}^{13}\}$. It significantly limits the problem complexity as the number of migration requests increases. The edge of the dependency graph indicates resource competition (network interfaces at source or destination, or bandwidth sharing along network routes) between migrations. 
A concurrent group equals to an independent set of the resource dependency graph.
A maximal concurrent scheduling group is a set of resource-independent migrations that is not a subnet of any other concurrent group. In other words, there is no other migration outside the concurrent group can be added to it so that all migrations can be performed at the same time. Therefore,  it equals a maximal independent set (MIS). The largest size MIS is a maximum independent set. As shown in Fig. \ref{fig: example-depgraph}, there are several combinations of migrations for a maximal concurrent scheduling group. In the first iteration, one of the maximum group is ${\left\{ {m_1^{13},m_7^{32},m_5^{25},m_9^{41}} \right\}}$ and one of the maximal group is $\left\{ {m_3^{23},m_{11}^{45},m_6^{31}} \right\}$. Thus, the maximum group is a better choice. After selecting the migrations from the nodes of the maximal independent set, we delete these migrations and update the dependency graph. One node is deleted from the graph when there is no migrations left in its migration list. For example, after the first iteration, we only delete nodes $v^{25},v^{41},v^{32}$, because there is still one migration $m_{12}^{13}$ left in node $v^{13}$ list. Thus, it is essential to select migrations carefully to achieve the maximum size of the concurrent groups.
At the end, the on-line scheduler schedules all migration in the first group. Then, when there is one migration finishes, the scheduler starts all migrations blocked by the finished migration following the order of migration groups.

Before discussing how to get the Maximum Independent Set, the largest Maximal Independent Set (MIS), of the resource dependency graph, we first review some basic graph concepts  \cite{bron1973algorithm}, such as clique $C$ and independent set $I$. A clique is a subset of vertices of an undirected graph G such that every two distinct vertices in the subset are adjacent. The maximal clique is a clique that cannot be extended by including one more adjacent vertex. On the other hand, an independent set of a graph G is the opposite of a clique that no two nodes in the set are adjacent. The maximum clique or independent set is the maximal clique or independent set with the largest size. $\alpha \left( G \right)$ denotes the size of the largest MIS of graph $G$. Therefore, an independent set of the resource dependency graph equals a concurrent migration group. The migrations from the nodes in an independent set $I$ can be scheduled concurrently. Meanwhile, migrations from the nodes in a clique are resource-dependent which need to be scheduled sequentially.

\subsection{Real-Time Planning}
There are few multiple migration planning and scheduling algorithms for live VM migration in cloud data centers \cite{bari2014cqncr, wang2017virtualold}.
However, the processing time of the scheduling sequence of multiple live migrations based on the algorithms in cloud data centers is not suitable for the real-time requirement of mobility-induced migrations in edge computing. For example, the processing time of FPTAS \cite{wang2017virtualold} and CQNCR \cite{bari2014cqncr} for migration planning is about 5 and 10 seconds for 100 migrations. The processing time increases to 44.56 and 968.46 seconds for FPTAS and CQNCR to generate the scheduling plan of 500 migrations. For traditional dynamic resource management, the algorithm triggered every 10 minutes or 30 minutes. This leaves enough time budget for algorithms to generate the optimal scheduling sequence. However, in the edge computing environment for mobility-induce live migrations, the live migration requests arrive at any time stochastically. Most of the migration requests are also time-critical. Thus, the processing time of the planning and scheduling algorithm for mobility-induced migrations should be adapted to suit the real-time scenario.

\subsection{Problem Modeling} \label{sect: model}
The planning and scheduling algorithm is triggered periodically after every time interval ${\Delta _{sch}}$.
We let $M_{arriv}^t$ denote the set of arrival migration requests at planning time $t$. $M_{wait}^t$ is the set of migration requests waiting for planning at time $t$. $M_{fail}^t$ is the set of infeasible migrations, such as its requested container is in migration. $M_{plan}^t$ is the set of migrations that have been planned but not finished at time $t$, and $M_{finish}^t$ is the set of finished migrations at time $t$.

The input of migration requests at every migration planning time $t$ is $M_{input}^t = M_{plan}^t \cup M_{wait}^t$. For each live container migration $m_j$, we have source and destination edge data center and allocated network routing, $(s_j, d_j, p_j)$, available bandwidth $l_j$, arrival time $a_j$, estimated migration time $T_j$, relative deadline $D_j$, start time $b_j$, and finish time $f_j$. Therefore, the response time can be represented as ${r_j} = {f_j} - a{}_j$. The slack time of migration scheduling, the remaining scheduling window that one migration will not miss its deadline, is ${\tau _j} = {a_j} + {D_j} - {T_j} - t$. 
The objective of live container migration planning and scheduling is to maximize the number of running resource-independent live migrations until the next planning time $t+ \Delta_{sch}$. 

At every planning and scheduling time $t$, the resource dependency graph
$G = (V, E)$ denotes the acyclic undirected graph where $|G| = |V|$. Each node $u \in |V|$ represents the list of migrations $M(u)$ where migration shares the same source $s$, destination $d$, and network routing $p$. By sharing the same source and destination and network routing, migrations in the list of a node are all resource-dependent. Let $(u,v) \in E$ denote the edge between node $u$ and $v$. It indicates the resource dependency between migrations from both nodes. 
$V\left( G \right)$ denotes the set of nodes of graph $G$.

We model the multiple migration planning problem as generating the maximal independent set of the dependency graph iteratively. In other words, in each iteration, we get the maximal independent set of the remaining graph, then update the graph by deleting corresponding migrations.
Let ${G_{i + 1}} = {G_i}\left[ {V\left( {{G_i}} \right) - {S_i}} \right]$ represent the remained graph by directly deleting vertex from set of nodes $S_i$.
Let $I_i$ denote the maximal independent set of graph $G_i$. Then, the remained graph $G_{i+1}$ in each iteration can be represented as:
\begin{equation}
{G_{i + 1}} = {G_i}\left[\kern-0.15em\left[ {V\left( {{G_i}} \right) - {I_i}} 
\right]\kern-0.15em\right] = {G_i}\left[ {V\left( {{G_i}} \right) - {S_i}} \right]
\end{equation}
by deleting set of nodes $S_i=\{u|u\in I_i,M(u) = \emptyset\}$, where the migration list of the deleted node $u$ is empty. 
Therefore, for each migration planning, the objective is to generate the iterative maximum independent set of dependency graph:
\begin{equation}
\max \left| {{I_i}} \right|,\forall {I_i} \in \left\{ {I_{iter}^i} \right\}
\end{equation}
where $\left\{ {I_{iter}^i} \right\} = \left\{ {{I_1},{I_2},...,{I_K}} \right\}$ is the total K iterative independent sets 
and there is no vertices left in the $K+1$ remaining graph as ${G_{K + 1}} = \emptyset$.
In other words, each iterative independent set size equals the size of maximum independent set of remaining graph $\left| {{I_i}} \right| = \alpha \left( G_i \right)$.

We extend the model to generate the iterative maximum weighted independent set for migration with different priorities, such as migration deadline.
The weight of an independent set is $W\left( I \right) = \sum\nolimits_{u \in I} {W\left( u \right)}$.
The largest weight of migration $\hat{m}$ in the node migration list is the weight of its corresponding node in the dependency graph $W\left( u \right) = W\left( \hat{m} \right)$ that
\begin{equation}
W\left( \hat{m} \right) \ge W\left( {m} \right),\forall \hat{m},m \in M\left( u \right)
\end{equation}
Then, the objective of multiple migration planning can be represented as:
\begin{equation}
\max W\left( I_i \right),\forall {I_i} \in \left\{ {I_{iter}^i} \right\}
\end{equation}
The weight of node $W(u) =1$ when there is no need to differentiate migrations in different nodes.
Generating the maximum (weighted) independent set of an undirected acyclic graph is a well known NP-hard problem \cite{lawler1980generating, cazals2008note}. Therefore, generating the iterative maximum independent set as the subset is also NP-hard.

\subsection{Complexity Analysis}
Because an independent set of $G$ is a clique in the complement graph of G and vice versa, the independent set problem and the clique problem are complementary \cite{bron1973algorithm, lawler1980generating, tomita2006worst}. In other words, listing all maximal independent sets or finding the maximum independent set of a graph equals listing all maximal cliques or finding the maximum clique of its complement graph. 
Thus, in each iteration, we can equivalently find the maximum independent set by getting the maximum clique ${C_i}$ of the complement graph ${C_i}\left( {{{\bar G}_i}} \right) = {I_i}\left( {{G_i}} \right)$.

It is known that all maximal cliques can be calculated in a total time proportional to the maximum number of cliques
in an n-vertex graph \cite{tomita2006worst}. In other words, each clique is generated in a polynomial time in all maximal cliques listing \cite{lawler1980generating}. When we only consider vertex, the maximal cliques listing algorithm (CLIQUES)  \cite{tomita2006worst,cazals2008note} based on Bron-Kerbosch \cite{bron1973algorithm} is the optimal algorithm.
The worst-case running time of CLIQUES is $O( {{3^{n/3}}} )$. The upper bound of all maximal cliques or independent sets of a graph is ${{3^{n/3}}}$ \cite{moon1965cliques}.
For the problem of finding one maximum independent set, the time complexity is improved from $O( {{2^{n/3}}} )$ in \cite{tarjan1977finding} to $O( {{2^{0.276n}}})$ \cite{robson1986algorithms}. Based on the work \cite{robson1986algorithms}, the best-known time complexity is $O\left( {{2^{n/4}}} \right)$ \cite{robson2001finding}. 
Therefore, it is computationally impossible to solve the exact problem of listing all maximal cliques (maximum clique) of its complement graph ${{\bar G}_{dep}}$ or all maximal independent sets (maximum independent set) of $G_{dep}$ for the real-time live container migration scheduling in edge computing which exhibits an exponential time complexity.

\section{Migration Planning and Scheduling} \label{section: multi-scheduling}
In this section, we present the proposed planning and scheduling algorithms for large-scale live container migrations in edge computing. 
With the waiting live container migration requests and planned unfinished live migrations as the input, the migration planner needs to efficiently schedule arriving migrations while maintaining the QoS.
Based on the problem modeling in Section \ref{sect: model}, this problem is reduced to finding an MIS of the migration dependency graph iteratively.
Therefore, we propose two major approaches to generate the iterative MISs of the dependency graph: (1) Direct iterative MIS generation and (2) Maximum Cliques (MCs)-based MIS generation.

\subsection{Direct iterative-Maximal Independent Sets}
For the direct iterative MIS generation, we follow the rationals based on the planning model as follows: (1) Create dependency graph $G_{dep}$ based on the source, destination, and network routing of the input migrations and the network topology; (2) Generate the Maximum Independent Set (MIS) $I$ of $G$; (3) Delete the nodes $u \in I$ from $G$ if its migration list $M(u)$ is empty; and (4) Repeat the procedure 2 and 3 until there is no vertices left $G_{dep} = \emptyset$.

\subsubsection{The Approximation}
For the approximation algorithm (approx) of creating the iterative maximum independent set, the procedure is as follows: 
In the approximation algorithm (Algorithm \ref{alg: approx}), we use the approximating maximum independent sets algorithm by excluding subgraph \cite{boppana1992approximating} to generate MIS in each iteration. 
Note that we skip the MIS generation and remove the migrations directly if the node size of $G_{dep}$ is unchanged in the current iteration. In other words, if we need to recalculate the MIS of the remaining graph, there is at least one node removed from the graph $G_i$. Given total $m$ live container migrations, we create the corresponding dependency graph with $n$ vertices. Therefore, regardless of the total number of migration requests, the upper bound of the complexity of planning multiple migrations scheduling is limited by the involved source, destination, and network routing. In the worst case, the planning algorithm only needs at most $n$ iteration rounds to calculate the concurrent migration group.
In each iteration, it guarantees $O(n/{(\log n)^2})$ approximate maximum independent set in polynomial time \cite{boppana1992approximating}. 

\begin{algorithm}[h]
	\caption{Iterative approximation grouping}\label{alg: approx}
	\KwIn{\{$G_{dep}$\}}
	\KwResult{migGroups $\{I_{iter}\}$}
	\SetKwFunction{appmis}{\textsc{Approx\_MIS}}
	$i \leftarrow 0$; ${G_i} \leftarrow G_{dep}$; $\left\{ {{I_{iter}}} \right\} \leftarrow \emptyset$;\\
	\While{$V\left( {{G_i}} \right) \ne \emptyset $}{
		$I_i \gets \appmis(G_i)$;\\
		${G_{i + 1}} \leftarrow {G_i}\left[\kern-0.15em\left[ {V\left( {{G_i}} \right) - {I_i}} 
		\right]\kern-0.15em\right]$;\\
		$\left\{ {{I_{iter}}} \right\} \leftarrow \left\{ {{I_{iter}}} \right\} \cup {I_i}$;\\
		$i \leftarrow i + 1$;\\
	}
\end{algorithm}

Based on the newly generated scheduling plan $\{I_{iter}\}$, the SDN-enabled on-line migration scheduler will start all feasible migrations in the first group $I_0$, considering the resource dependency with current running migrations. Then, whenever a migration finishes, the scheduler starts all remaining feasible migrations in each concurrent migration group $I_i$ followed by the scheduling plan order.

\subsubsection{Greedy MIS Algorithm}
The greedy algorithm (iter-GWIN) generates the concurrent groups (MIS) of live migration iteratively.
A greedy maximal independent set algorithm (GWIN) \cite{sakai2003note} based on the weight and the degree of a node is adapted to directly generate the MIS in each iteration.
Let ${\Delta _G}$ denote the maximum degree and ${{\bar d}_G}$ is the average degree of $G$. The degree of node $u$ in $G$ is ${d_G}\left( u \right) = \left| {{N_G}(u)} \right|$.
${N_G}\left( u \right)$ is the set of neighbor nodes of vertex $u$ and $N_G^ + \left( u \right) = {N_G}\left( u \right) \cup \left\{ u \right\}$.

\begin{algorithm}
	\caption{iter-GWIN} \label{alg: iter-gwin}
	\KwIn{\{$G_{dep}$\}}
	\KwResult{migGroups $\{I_{iter}\}$}
	$i \leftarrow 0$; ${G_i} \leftarrow G_{dep}$; $\left\{ {{I_{iter}}} \right\} \leftarrow \emptyset$;\\
	\While{$V\left( {{G_i}} \right) \ne \emptyset $}{
		${I_i} \leftarrow \emptyset$; ${G_j} \leftarrow {G_i}$; $j \leftarrow 0$;\\
		
		\While{$V\left( {{G_j}} \right) \ne \emptyset $}{
			select node $\hat{u}$ in $G_j$;\\
			${I_i} \leftarrow {I_i} \cup \left\{ {\hat u} \right\}$;\\
			${G_{j + 1}} \leftarrow {G_j}\left[ {V\left( {{G_j}} \right) - N_G^ + \left( {\hat u} \right)} \right]$;\\
			$j \leftarrow j + 1;$\\
		}
		${G_{i + 1}} \leftarrow {G_i}\left[\kern-0.15em\left[ {V\left( {{G_i}} \right) - {I_i}} 
		\right]\kern-0.15em\right]$;\\
		$\left\{ {{I_{iter}}} \right\} \leftarrow \left\{ {{I_{iter}}} \right\} \cup {I_i}$;\\
		$i \leftarrow i + 1$;\\
	}
\end{algorithm} 

As shown in Algorithm \ref{alg: iter-gwin}, from line 3-8, it selects the node with largest score regarding the minimal degree and maximal weight:
\begin{equation} \label{eq: node-score}
{{W\left( u \right)} \mathord{\left/
		{\vphantom {{W\left( u \right)} {\left( {{d_{{G_i}}}\left( u \right) + 1} \right)}}} \right.
		\kern-\nulldelimiterspace} {\left( {{d_{{G_i}}}\left( u \right) + 1} \right)}}
\end{equation}
It removes the selected node and its neighbors from the graph and repeats the procedure until there is no vertices left. 

As mentioned in the problem modeling, the weighted node equals the maximum weight of migrations from its list. 
The migration weight could be arrival time, estimated migration time, or correlation network influence \cite{bari2014cqncr} after migration for non-time-critical migrations and the deadline or slack time for real-time migrations scheduling. In this paper, we consider the weight function regarding the slack time $\tau$ as follows:
\begin{equation} \label{eq: weight}
W\left( m \right) = \left\{ {\begin{array}{*{20}{c}}
	{{{10 \cdot \beta } \mathord{\left/
				{\vphantom {{10 \cdot \beta } \tau }} \right.
				\kern-\nulldelimiterspace} \tau }}\\
	{{{100 \cdot \left| \tau  \right|} \mathord{\left/
				{\vphantom {{100 \cdot \left| \tau  \right|} \beta }} \right.
				\kern-\nulldelimiterspace} \beta }}\\
	{100}
	\end{array}\begin{array}{*{20}{c}}
	{\tau  > \beta }\\
	{\tau  <  - \beta }\\
	{other}
	\end{array}} \right.
\end{equation}
where $\beta$ is the slack time threshold. We set $\beta = 1$ in this paper.
The weight of node is  $W(u) = \gamma \cdot W(m)$, where $\gamma$ is the coefficient regulator for the urgency of the scheduling migration. We set $\gamma = 1$.
Moreover, in the situation that the priorities of all migrations are the same, we only need to consider the size of MIS. The node weight is set to 1 $W(u) =1$.
In each iteration, the lower-bound of the maximum (weighted) independent set is $\sum\nolimits_{u \in V} {{{W\left( u \right)} \mathord{\left/
			{\vphantom {{W\left( u \right)} {\left( {{d_G}\left( u \right) + 1} \right)}}} \right.
			\kern-\nulldelimiterspace} {\left( {{d_G}\left( u \right) + 1} \right)}}} $ \cite{sakai2003note}.
As iteration is $n$ in the worst case,
the time complexity of iter-GWIN is $O({n^2}\log n)$ for weighted graph and $O\left( {n^2} \right)$ for unweighted graph.

\subsection{The Maximum Cliques-based Heuristics} \label{section: iter-mcs}
In this section, based on the observation of the density property of migration resource dependency graph, we propose the iterative Maximum Cliques (MCs)-based algorithm. We first discuss the rationals of the proposed algorithm.

The degeneracy of a graph G is the smallest number d such that every subgraph of G contains a vertex of degree at most d. It is a measure for the graph spareness.
For an n-vertex graph with degeneracy $d$, by introducing the sequence ordering based on degeneracy, Bron-Kerbosch
Degeneracy algorithm \cite{eppstein2010listing} can list all maximal cliques in time $O( {dn{3^{d/3}}} )$.
With a spare graph that $n \ge d + 3$, the upper bound of all maximal cliques number is $\left( {n - d} \right){3^{d/3}}$.
Figure \ref{fig: graph-density} illustrates the nodes and the density (degeneracy) of the resource dependency graph of WAN network topologies \cite{knight2011internet} and its complement.
It shows that the degeneracy of the complement graph ${{\bar G}_{dep}}$ is 4.34 times that of $G_{dep}$. For $G_{dep}$ and its complement graph, the average ratio of dependency $d$ to the total number of nodes $n$ is 0.153 and 0.714, respectively. The resource dependency graph is considerably more sparse than its complement graph. 
Therefore, for $G_{dep}$, there are much fewer maximal cliques than the total MIS. 
As a result, according to the theoretical time complexity, the running time of listing all maximal cliques or maximum clique of $G_{dep}$ is much smaller than that of listing all maximal independent sets or maximum independent set of $G_{dep}$.
Therefore, the iterative Maximum Cliques (MCs)-based heuristics algorithm has two steps: (1) calculate the list of iterative maximum cliques and (2) generate the iterative maximal independent set based on the list. As nodes from one maximal clique can not be included into the same independent set, the iterative maximum cliques serve as a heuristic pruning decider to speed up the algorithm.

\begin{figure}[t]
	\centering
	\includegraphics[width=0.8\linewidth]{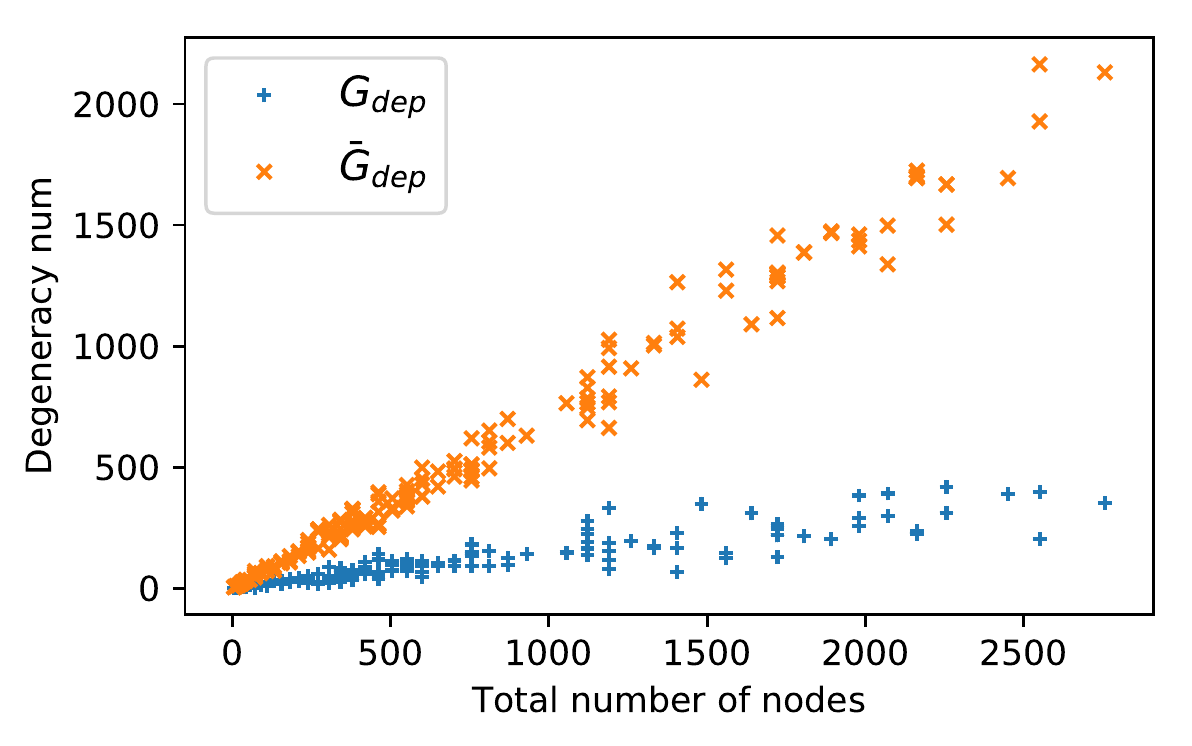}
	\caption{Total number of nodes and degeneracy number of resource dependency graph of WAN topologies \cite{knight2011internet} and its complement graph}
	\label{fig: graph-density}
\end{figure}

\subsubsection{Iterative-rounds MCs algorithm}

\begin{algorithm}[h]
	\caption{Iterative heuristic of migration grouping}\label{alg: heur-iter}
	\KwIn{\{$G_{dep}$\}}
	\KwResult{migGroups $\{I_{iter}\}$}
	\SetKwFunction{maxc}{\textsc{Maximum\_Clique}}
	\SetKwFunction{sort}{\textsc{SortNode}}
	\SetKwFunction{addIndep}{\textsc{addIndep}}
	\SetKwFunction{delnode}{\textsc{DelNode}}
	$\left\{ {{C_{iter}}} \right\} \leftarrow \emptyset$;$\left\{ {{I_{iter}}} \right\} \leftarrow \emptyset$;\\
	\While{$|G_{dep}|$ $!=$ $0$}{	\{Iterative creating Maximum Cliques\}\\
		$\hat {C_{i}}  \gets \maxc(G_{dep})$;\\	
		${G_{dep}} \leftarrow {G_{dep}}\left[ {V\left( {{G_{dep}}} \right) - {{\hat C}_i}} \right]$;\\
		$\{C_{iter}\} \gets \{C_{iter}\} \cup {{\hat C}_i}$;\\
	}
	
	\While{$\{C_{iter}\} \ne \emptyset$ }{
		$I \gets \emptyset$ \\
		\ForEach{ ${{\hat C}_i}$ in $\{C_{iter}\}$}{
			$m \gets \addIndep(I, {{\hat C}_i})$;\\
			\textsc{\delnode(${{\hat C}_i}$, $m$)};\\
		}
		$\{I_{iter}\} \gets \{I_{iter}\} \cup I$;\\
	}
	\KwRet{$\{I_{iter}\}$}
\end{algorithm}

Let ${{\hat C}_i}$ denote the maximum clique and $\left\{ {{{\bar C}_i}} \right\}$ denote the maximal cliques list of round $i$ graph. The iterative-rounds Maximum Cliques (MCs)-based heuristic algorithm (Algorithm \ref{alg: heur-iter}) follows two steps: (1) generating the maximum clique iteratively and (2) obtaining the MIS from the iterative maximum cliques.

As shown in Algorithm \ref{alg: heur-iter}, we first create dependency graph $G_{dep}$ as the input based on the source-destination of the given migrations and the network topology. 
From line 1-6, the algorithm calculates the iterative maximum cliques of the dependency graph until there is no vertices left.
In each iteration, it generates the maximum clique (Bron-Kerbosch Degeneracy algorithm) \cite{eppstein2010listing} of the remaining graph. It is proved that the algorithm is highly efficient in a sparse graph, such as the resource dependency graph \cite{eppstein2010listing}. Then, it updates the remaining graph by deleting the nodes of the maximum clique from $G_{dep}$. 
Let $d_{G[C]}(u) = |N_{G[C]}(u)|$ denote the degree of node $u$ to the remaining graph which excludes all nodes in the clique. 
The node score can be represented as:
\begin{equation}
{{W\left( u \right)} \mathord{\left/
		{\vphantom {{W\left( u \right)} {\left( {{d_{G\left[ C \right]}}\left( u \right) + 1} \right)}}} \right.
		\kern-\nulldelimiterspace} {\left( {{d_{G\left[ C \right]}}\left( u \right) + 1} \right)}}
\end{equation}

In the second step (line 7-12), it generates maximal independent sets based on the iterative maximum cliques. 
In each round (line 9-11), it selects the feasible node with maximum score of each maximum clique ${{\hat C}_i}$ and adds largest-weight migration from its list into the independent set. A node is feasible when it can be included in the current independent set. If there is no migrations left in the migration list of the selected node $M(u)$, the selected node is removed from the clique.
As the largest possible number of maximal cliques in an n-vertex graph with degeneracy d is 
$\left( {n - d} \right){3^{d/3}}$. Therefore, according the iter-MCs algorithm, the upper bound of the size of the iterative maximum independent set of each iteration is also $\left( {n - d} \right){3^{d/3}}$. In the worst case, the time complexity of iter-MCs is $O( {dn^2{3^{d/3}}} )$.

\begin{thm}[Correctness of MIS from Maximal Cliques]
	The Independent Sets generated from maximal cliques are the maximal independent sets of the graph.
\end{thm}

\begin{proof}
	${I_q} = \left\{ {{q_0},{q_1},{q_2},...,{q_d}} \right\}$ is one of the independent sets generated from the maximal cliques of $G\left( {V,E} \right)$, 
	where one vertex comes from only one maximal clique $q \in {C_q}$. 
	Assume, for the sake of contradiction, there is at least one vertex $p$, $p \in {C_p}$ exists, that ${I_q} \cup \left\{ p \right\}$ is also an independent set.
	That is, there is no edge between $p$ and any other vertex $\forall q, q \in {I_q}$, $\neg \exists \left( {p,q} \right) \in E$. 
	Based on the definition of the heuristic algorithm, we can get $\forall r \in {C_p}$, $r \notin {I_q}$, that $\exists q \in {I_q}$, where $\left( {p,r} \right) \in E$.
	Thus, $\exists p,q$, where $p \in {C_p},q \in {I_q}$, that $\neg \exists \left( {p,q} \right) \in E$ and $\exists \left( {p,q} \right) \in E$, which is impossible.
	Since, we have a contradiction, it must be that ${I_q}$ is a maximal independent set.
\end{proof}

\subsubsection{Single-Round MCs Algorithm}
Furthermore, we propose a single-round MCs-based algorithm (single-MCs). It generates the optimal iterative maximum cliques only based on the all maximal cliques of the initial dependency graph $G_{dep}$. The maximum clique size of each iteration is the same as the iter-MCs. We also prove the correctness of the proposed single-round iterative maximum cliques algorithm.

\begin{algorithm}[h]
	\caption{Single-round iterative maximum cliques}\label{alg: heur-single}
	\KwIn{\{$G_{dep}$\}}
	\KwResult{migGroups $\{C_{iter}\}$}
	\SetKwFunction{maxc}{\textsc{Maximum\_Clique}}
	\SetKwFunction{sort}{\textsc{SortDegree}}
	\SetKwFunction{addIndep}{\textsc{addIndep}}
	\SetKwFunction{delnode}{\textsc{DelNode}}
	\SetKwFunction{fcliq}{\textsc{FindCliques}}
	\SetKwFunction{siniter}{\textsc{SINGLE-ITER}}
	\SetKwFunction{delnodes}{\textsc{DelNodes}}
	\nonl\siniter{$\{G_{dep}\}$}:\\
	$\left\{ {{C_{iter}}} \right\} \leftarrow \emptyset$;\\
	$\left\{ {{{\bar C}}} \right\} \gets \fcliq(G_{dep})$;\\
	\While{$\left\{ {{{\bar C}_i}} \right\}$ $\ne \emptyset$}{
		${{\hat C}_i} \gets \textsc{max}(\left\{ {{{\bar C}}} \right\})$;\\
		$\{C_{iter}\} \gets \{C_{iter}\} \cup {{\hat C}_i}$;\\
		$\left\{ {{{\bar C}}} \right\}$ $\gets$ \delnodes(${{\hat C}_i}$, $\left\{ {{{\bar C}}} \right\}$);\\
	}
	\KwRet{$\{C_{iter}\}$}
\end{algorithm}

The first step of the iter-MCs algorithm is replaced by Algorithm \ref{alg: heur-single}.
The algorithm only generates the list of all maximal cliques $\{\bar{C}\}$ once by using the Bron-Kerbosch Degeneracy algorithm. 
Until there is no vertices left in the clique list, it selects the maximum clique (largest maximal clique) $\hat{C}_i$ from the list and deletes the nodes of the selected maximum clique from all maximal cliques. 

\begin{thm}[Correctness of the algorithm single-MCs]
	Given a graph $G=(V,E)$ $V \ne \emptyset$ , the single iteration algorithm SINGLE-MCs generates all and only iteration maximum cliques. 
\end{thm}

\begin{proof}
	It is proven that the Bron-Kerbosch Degeneracy algorithm generates all and only maximal cliques without duplications \cite{eppstein2010listing}. Then, we only need to prove the results of iterative maximum cliques are the same in iter-MCs and single-MCs, i.e., one can get all the iterative maximum cliques based on the maximal cliques of the original graph by deleting the vertices from the maximum clique in the last round.
	
	Let $C\left( G \right) = \left\{ {{C_0},{C_1},...,{C_d}} \right\}$ denote all maximal cliques of the original graph $G$, where $\left| {{C_i}} \right| \ge \left| {{C_{i + 1}}} \right|$.  $\forall {C_i},{C_j} \in C\left( G \right)$, that ${C_i} \ne {C_j},{C_i} \not\subset {C_j}$. The next iteration graph is $G\backslash {C_0} = G\left[ {V\left( G \right) - {C_0}} \right]$. Then, $C(G\backslash {C_0}) = \{ C_1^{'},C_2^{'},...,C_e^{'}\}$. The output of first round of single-MCs is $C\left( G \right)\backslash {C_0} = \{ C_1^{''},C_2^{''},...,C_e^{''}\}$. 
	
	Assume, for the sake of contradiction, there is one maximal clique  ${C_f} = C_i^{''} \cup \left\{ q \right\}$, $q \in V - {C_0},q \notin C_i^{''}$, which ${C_f} \in \left\{ {C_j^{'}} \right\}$ and ${C_f} \notin \left\{ {C_i^{''}} \right\}$. Based on the algorithm single-MCs and definition of maximal clique, due to $\left\{ q \right\} \notin {C_0}$,  $C_f \cup {C_0}$ is also a maximal clique that ${C_f} \cup {C_0} \in C\left( G \right)$. However, as $C_0$ is the maximum clique of $G$, it is impossible that ${C_f} \notin \emptyset$. Since, we have a contradiction, the $C_f$ is not exist. Therefore, $C\left( G \right)\backslash {C_0} = C(G\backslash {C_0})$.
	
\end{proof}

\section{Graph Algorithm Performance and Analysis} \label{section: evaluation}

In this section, we evaluate proposed migration planning algorithms for the problem of iterative MIS generation: (1) iter-MCs (2) single-MCs; (3) approximation; and (4) iter-GWIN, in processing time, maximal independent set size, and iteration rounds.
Based on more than two hundred real network WAN topologies \cite{knight2011internet}, we consider a set of live migration requests with each source and destination combination. Each migration request corresponds to one combination with the network routing of the shortest path. We run the computational experiments in Python 3.6.3 and NetworkX package \cite{hagberg2008exploring} version 2.4 as the graph library with source code. 

\begin{table*}[t]
	\caption{Performance comparison with first, second, and third quartile of processing time, total MIS number (iteration), maximum, mean, 95th, and 99th quartile of the independent set size in each result of the total 202 WAN topologies}
	\label{tb: alg-evaluation}
	\centering
	\resizebox{\linewidth}{!}{%
		\begin{tabular}{|l|l|l|llll|}
			\hline
			algorithm		&	proc. time (s)												&	total sets $|\{I\}|$		&	max($\{|I|\}$)  &	mean($\{|I|\}$)			 			&	95\%($\{|I|\}$)		&  99\%($\{|I|\}$)		\\
			\hline
			single-MCs 		&	8.8115				1.0047				0.1554				&	164.5 75.0 30.0				&	88.0 54.0 36.0	&			8.1594 6.1667 5.0			&	23.05 16.8 12.825	&	45.9 25.8 17.7			\\
			iter-MCs 		&	49.1566				4.5807				0.4723				&	165.0 75.0 30.0				&	88.0 54.0 36.0  &			8.1594 6.2124 5.0  		   	&	23.0 16.8 12.65		&	45.6 26.0 17.7			\\
			iter-GWIN 		&	14.2786				1.4916				0.1610				&	159.5 76.0 31.0				&	88.0 54.0 36.0  &			8.2844 6.3448 5.1909  		&	24.8 18.0 13.15		&	48.9 27.0 18.6			\\
			approx			&	1115.2929			57.2547				5.5257				&	171.5 84.0 32.0				&	59.0 36.0 25.0	&			7.7568 5.9492 4.6946		&	21.6 15.85 12.0		&	36.8 23.2 15.8			\\
			\hline
		\end{tabular}
	}
\end{table*}

\begin{figure}[t]
	\centering
	\includegraphics[width=0.9\linewidth]{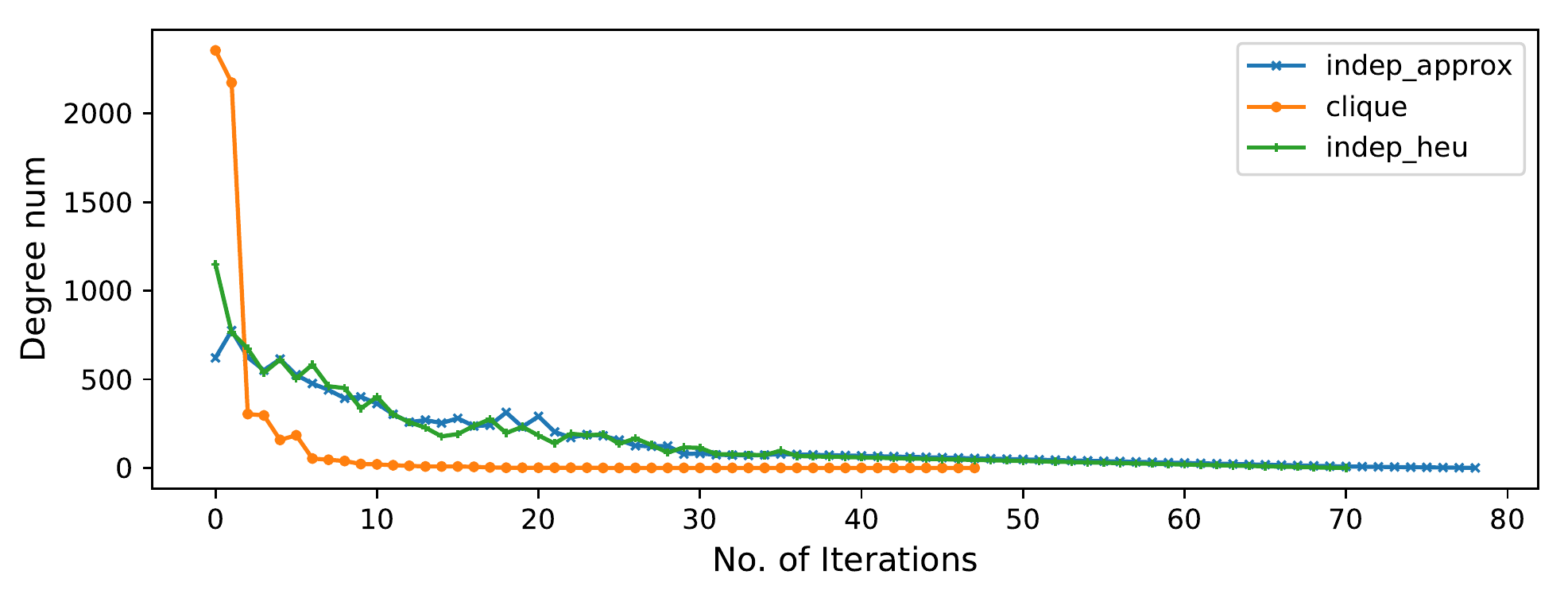}
	\caption{Degree of iteration cliques/independent set to the remaining nodes}
	\label{fig: depGraph-degree}
\end{figure}

The iterative maximum clique generation of migration dependency graph is faster than that of iterative MIS in three aspects: (1) As the analysis of dependency graph in section \ref{section: iter-mcs}, dependency graph $G_{dep}$ is more sparse than its complement $\bar{G}_{dep}$. Therefore, getting the maximum clique of $G_{dep}$ in one iteration is faster than that of $\bar{G}_{dep}$; 
(2) The maximum clique can reduce the complexity of the graph much more efficiently in each iteration; and (3) There are fewer iterative maximum cliques of $G_{dep}$ than the iterative independent sets. In other words, the number of iteration rounds of the iterative maximum clique is smaller.

Figure \ref{fig: depGraph-degree} demonstrates an illustrative result of one of the network topology (Australia's Academic and Research Network, AARNet). The dependency graph consists of a total of 342 nodes and 11754 edges. It shows the degree of the maximum clique and independent set in each iteration to the remaining nodes. In other words, it is the edges of the removed nodes in each iteration excluding the edges between nodes from the maximum clique. Note that there is no edges (degree is zero) between nodes in one independent set.
With the degree in the maximum clique and the degree to the remaining graph, the complexity of $G_{dep}$ is dropped dramatically in the first three iterations. On contrary, by removing the maximum independent set, the complexity of the graph remains at a high level and declines steadily. Furthermore, the number of total iterative maximum cliques and iterative MISs is 47 and 70, respectively.
Comparing the result of the approximation (indep\_approx) with iter-MCs (indep\_heu), the heuristic iterative MCs-based algorithm achieves better performance in the size of the maximum clique in each iteration and the total iteration rounds.

\begin{figure}[t]
	\centering
	\includegraphics[width=0.9\linewidth]{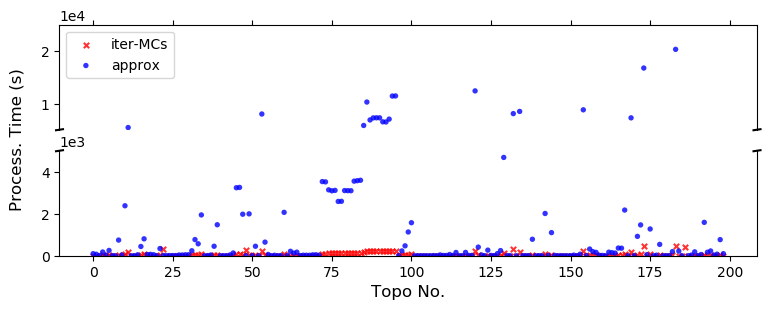}
	\caption{Processing time comparison between iter-MCs and approximation}
	\label{fig: proc-time1}
\end{figure}

\begin{figure}[t]
	\centering
	\includegraphics[width=0.9\linewidth]{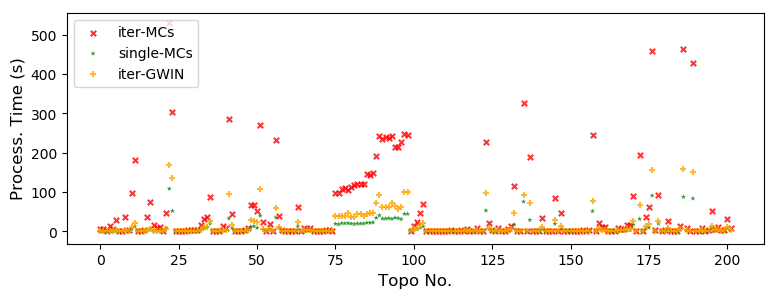}
	\caption{Processing time comparison between iter-GWIN, iter-MCs, and single-MCs}
	\label{fig: proc-time2}
\end{figure}

Since the processing time varies greatly, we use two separated figures to represent the results of processing time.
Figure \ref{fig: proc-time1} shows the performance comparison between the approximation (approx) and iter-MCs in processing time. Figure \ref{fig: proc-time2} shows the comparison between iter-MCs, single-MCs and iter-GWIN.
The results of computational experiments indicate that the approximation algorithm has the worst performance in processing time. From approx to iter-MCs (Fig. \ref{fig: proc-time1}), the average processing time of all topologies decreases by $91.32\%$. From iter-MCs to iter-GWIN and iter-GWIN to single-MCs, the average processing time decreases by $57.40\%$ and $20.84\%$.
Table \ref{tb: alg-evaluation} also illustrates the third (Q3), second (mean), and first quartile (Q1) of the average processing time. For the dependency graph with every source and destination combinations of a relative small size network, the single-MCs and iter-GWIN can both generate the scheduling plan in around 0.15 seconds. However, the performance difference in processing time increases with the size of the network topology. For mean and the Q3 of all processing time results, the average processing time of single-MCs decreased by $32.64\%$ and $38.29\%$ from iter-GWIN, respectively. 
In summary, the single-MCs algorithm has the best performance in processing time. 

We also evaluate the size of the result list or iteration rounds $|\{I\}|$. It is the number of sets the algorithm divides into different concurrent groups for the given migrations. For the approx algorithm, from Q3 to Q1, it generates 171.5, 84.0, and 32.0 many of iterative MIS in one planning result.
From approx to iter-MCs, the iteration number decreases by $3.79\%$, $10.71\%$, and $6.25\%$, respectively.

For the performance in iterative MISs of each graph, we examine the size of the largest iterative MIS (max($\{|I|\}$)) and the mean size (mean($\{|I|\}$)). As the first several rounds of the result are the most essential factors on scheduling performance, we also evaluate the algorithm in the 95-quartile and 99-quartile of the iterative MISs size.
The algorithm iter-GWIN has the best performance in the large network topology. The total number of iterative MIS is reduced by $3.04\%$ compared to the results of single-MCs.
Although the mean results of the set size mean($\{|I|\}$) of approx algorithm is close to other three algorithms, its performance in the first several iterations is the worst. As a result, the total set of approx algorithm is significantly larger than other algorithms.
For the maximum set size, single-MCs, iter-MCs and iter-GWIN has the identical performance in Q1, mean, and Q3 from all results of network topologies. For the 95th and 99th quartile iter-GWIN for directly calculate the maximum clique has a slightly better performance over the iterative MCs-based heuristic algorithms even though the processing time is higher.

\section{Simulation and Performance Evaluation} \label{section: experiment}
In this section, we evaluate proposed solutions using real-world traces on an event driven simulator.
We first describe the real-world telecom base station dataset and taxi GPS traces used in the experiments. 
We explain the placement of edge data centers and the network topology and region coverage of each EDC. 
The event-driven simulator for software-defined network-enabled edge-cloud computing CloudSimSDN \cite{son2019cloudsimsdn} is extended to emulate the the user movement and the live container migration in edge computing. It provides a network operating system based on the software-defined networking for dynamic service and network resource monitoring and allocation. Compared to the simulation results driven by mathematical models, this can generate more realistic results without following the strong assumption encoded in the proposed mathematical modeling.

We compare and evaluate the performance of live container migration planning and scheduling algorithm (iter-GWIN and single-MCs)
against a policy with no planning scheduling and the state-of-art live VM migration cloud algorithm FPTAS~\cite{wang2017virtualold} in processing time, migration time, downtime, transferred data, deadline violations, and network transmission time.

\begin{figure*}[th]
	\centering
	\begin{subfigure}{.20\linewidth}
		\centering
		\includegraphics[width=\linewidth]{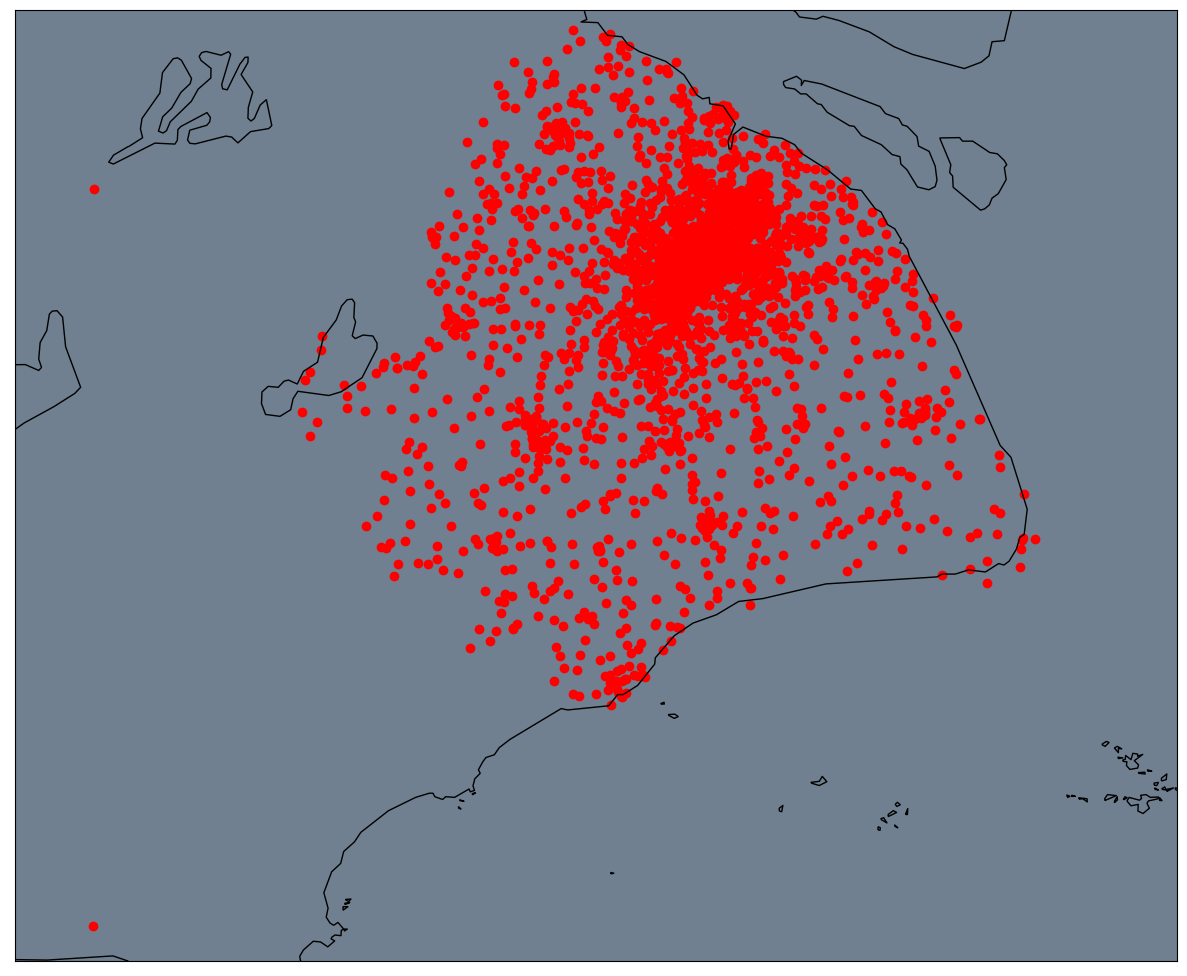}
		\caption{Shanghai Telecom base station locations}
		\label{fig: base-station}
	\end{subfigure}%
	\hspace{0.1em}
	\begin{subfigure}{.24\linewidth}
		\centering
		\includegraphics[width=\linewidth]{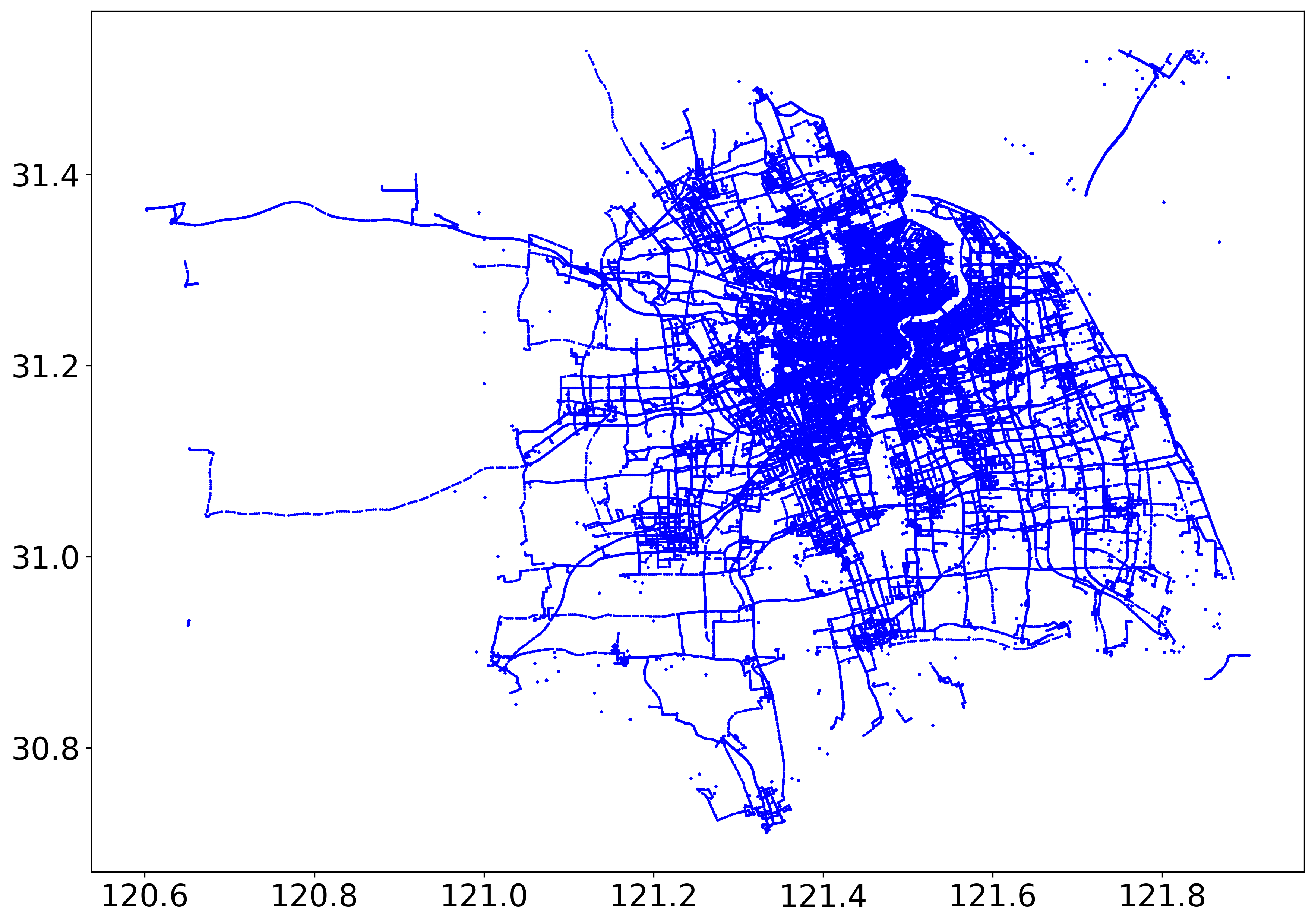}
		\caption{The taxi GPS trace in the first hour}
		\label{fig: taxi-dataset}
	\end{subfigure}
	\hspace{0.1em}
	\begin{subfigure}{.24\linewidth}
		\centering
		\includegraphics[width=\linewidth]{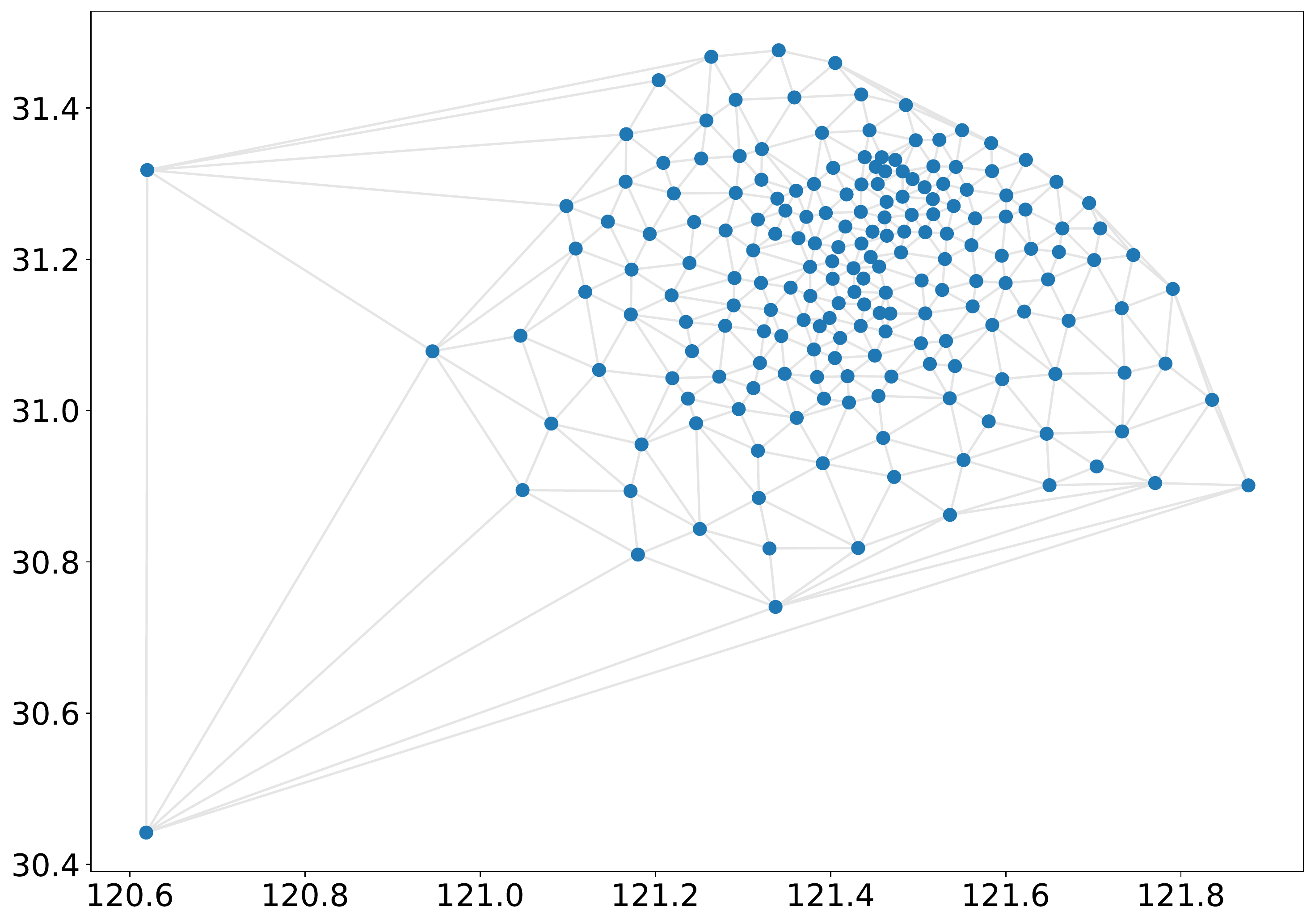}
		\caption{Delaunay triangulation based Edge computing topology}
		\label{fig: edge-topo}
	\end{subfigure}
	\hspace{0.1em}
	\begin{subfigure}{.24\linewidth}
		\centering
		\includegraphics[width=\linewidth]{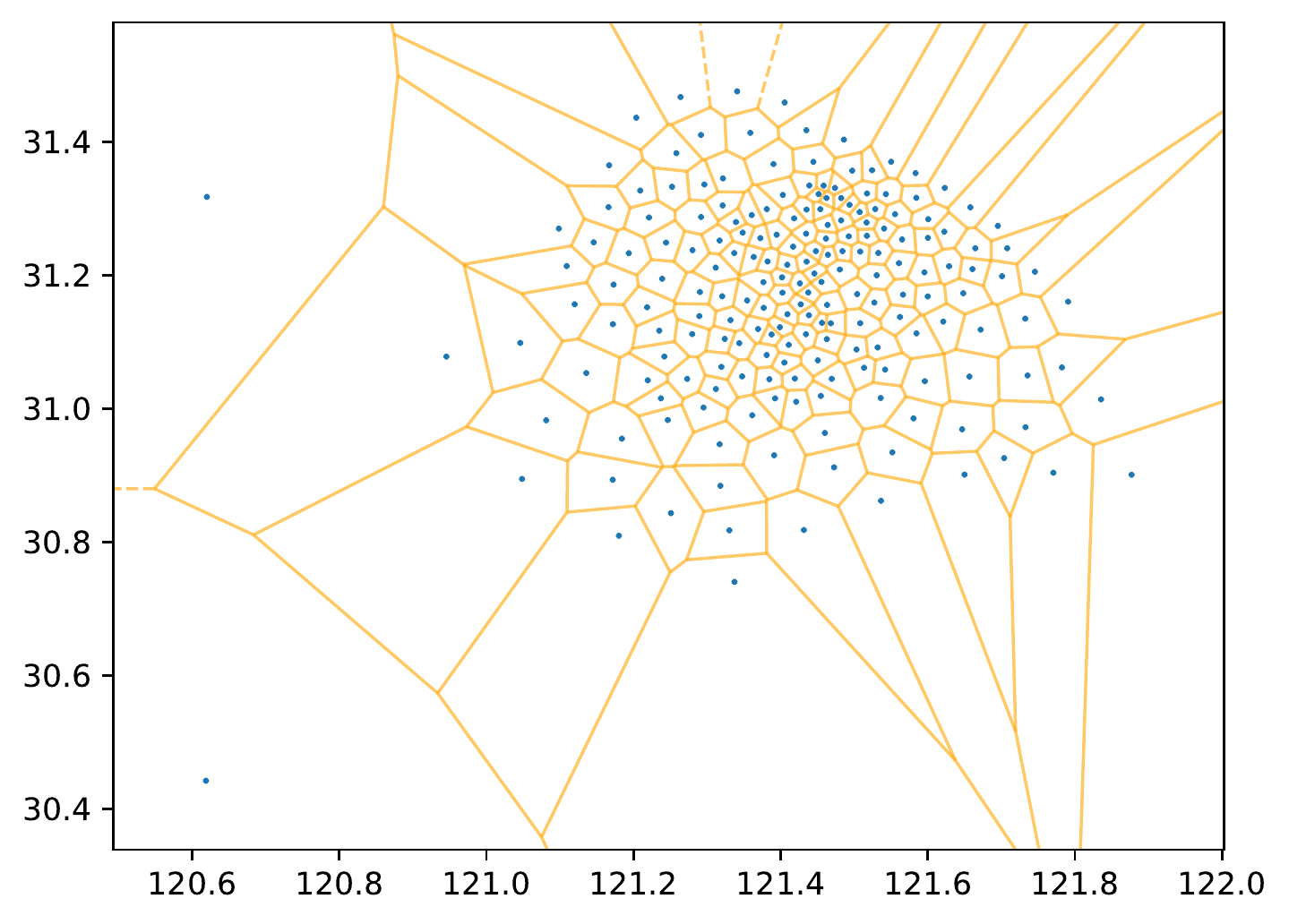}
		\caption{Edge data center coverage with Voronoi cells}
		\label{fig: edge-coverage}
	\end{subfigure}
	\caption{Experimental dataset and configurations of longitude and latitude}
\end{figure*}

\subsection{Experimental Data}
In this section, we describe the base stations coordinates provided by Shanghai Telecom dataset\footnote{http://sguangwang.com/TelecomDataset.html} and Shanghai Qiangsheng taxi GPS trace dataset (April 1, 2018)\footnote{http://soda.shdataic.org.cn/download/31} used in our experiments. 
The given data is preprocessed by limiting the range of the longitude and latitude from ${30.40^ \circ }$ N to ${31.35^ \circ }$ N and ${120.51^ \circ }$ E to ${122.12^ \circ }$ E as there are some taxis travel to nearby cities.
The Shanghai Telecom dataset contains the longitude and latitude coordinates of a total of 3233 base stations as shown in Fig. \ref{fig: base-station}. 
We use K-means algorithms \cite{xu2015efficient} to generate the location of a total of 200 Edge Data Centers (EDCs) based on the longitude and latitude of the given base stations. Figure \ref{fig: taxi-dataset} illustrate the taxi GPS trace in the first hour. Besides the GPS coordinates and timestamp, each data record of the taxi dataset also includes the taxi id, service status, such as alarm, occupation, taxi light, road type, and breaking, as well as vehicle speed, direction, and the number of connected satellites.

Figure \ref{fig: edge-topo} illustrates base stations are clustered and connected to one of the regional Edge Data Centers. There is no information on the physical network topology and connectivity between EDCs. As shown in Fig. \ref{fig: edge-topo}, for the geometric spanner, we choose Delaunay Triangulation \cite{virtanen2020scipy} to generate links between the gateway of each EDC. For the network routing within the generated network topology, we consider the shortest path, which is no longer than 
$4\pi/3\sqrt{3}$
times the Euclidean distance between source and destination. As a result, the boundary of EDC regions (Fig. \ref{fig: edge-coverage}) is a Voronoi diagram \cite{virtanen2020scipy} where the Euclidean distance of any point to its corresponding EDC region is less than or equal to its distance to any other EDC.

\begin{figure}[ht]
	\centering
	\includegraphics[width=0.8\linewidth]{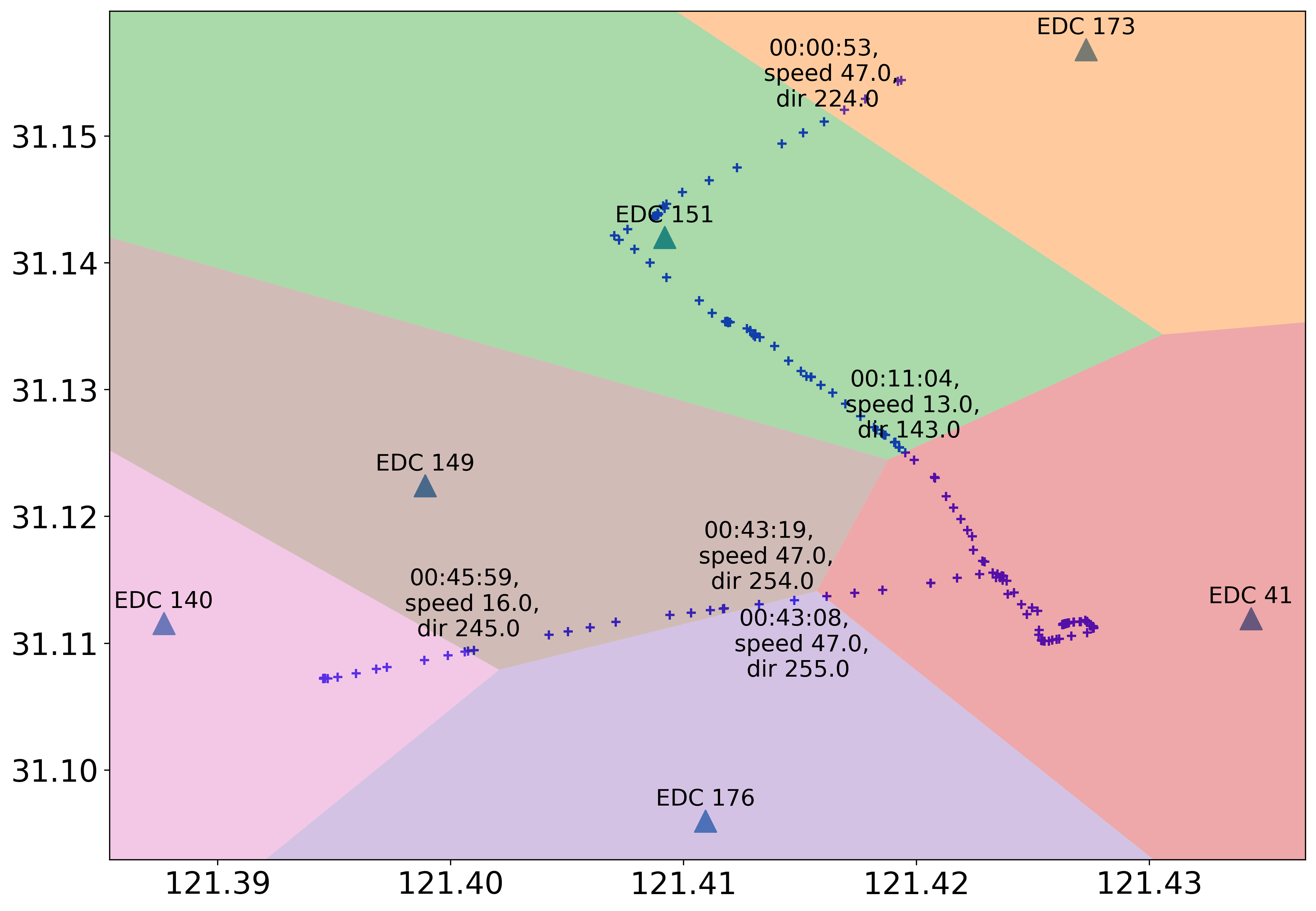}
	\caption{Example of live migration request triggered by user movement of longitude and latitude}
	\label{fig: user-move-example}
\end{figure}

Similar to other research regarding the generation of mobility-induced live migrations in edge computing \cite{wang2019dynamic,wang2019delay}, we combine these two datasets to simulate the scenarios where the user needs to connect to the services and maintains the low end-to-end latency through live container migration in edge computing environments.
Figure \ref{fig: user-move-example} demonstrates an example that the request of live container migration is induced when a taxi moves across the boundary between two clusters of EDCs. The deadline of each live container migration is generated based on the average mobility speed of users in the last 3 GPS records, travel direction, and the signal strength of base stations.

\subsection{Experimental Setup}

\begin{table}[ht]
	\caption{Multiple container migration scenarios}
	\label{tb: scenarios}
	\centering
	\begin{tabular}{|l|lll|}
		\hline
		Scenario  & S1 		 & S2 		& S3 \\
		\hline
		vehicles  & 1000      &   2000   & 4000	 \\
		migrations &  9933    &   19522    &37822 		\\
		\hline
	\end{tabular}
	
\end{table}

In this section, we describe details of the experiment setup.
The end-to-end delay between the user and the service is the time interval from the user (taxi) sent workload to the container assigned in the EDC to the result is received by the user.
To generalize computer vision use case workloads, the service task generated during the experiments follows the Poisson distribution with a mean of 24 per second (24 FPS).
In each task, the network packet size sent from a user is 16384 bytes ($128*128$ bytes). The processing workloads in the container are randomly generated from 500 to 1000 cycles per bit \cite{wang2019delay}. The result packet sent back from the container to the user is 128 bytes. The sum CPU power frequency for each EDC with multiple CPUs is 25 GHz \cite{sun2017emm,xiao2017qoe}. 
To simulate the limited network resources for migrations in the edge, we consider that the reserved network bandwidth between a container and its user is 3 Mbps. The physical network bandwidth is 1 Gbps. The network delay between any based station to its regional EDC is 5 ms and the delay between two EDCs is randomly generated in the range 5 to 50 ms~\cite{wang2019delay}. 

According to the evaluation results of container memory and dirty memory size during live container migrations \cite{ma2018efficient}, we generate the container memory from 100 MB to 400 MB.
The dirty page rate for each dirty memory transmission is from 2MB/s to 8MB/s and the data compression rate is 0.8 \cite{svard2011evaluation}. We configure the downtime threshold and the maximum iterations for live container migration at 0.5 seconds and 30 times~\cite{he2019performance}, respectively. Based on the SDN controller, the remaining network bandwidth between the source and destination EDC which is not utilized by services is allocated to the live container migration traffic. If several live migrations are sharing part of their routings, the bandwidth will be allocated evenly to each of the network flows.

From the experimental scenario S1 to S3 (Table \ref{tb: scenarios}), 1000, 2000 and 4000 vehicles are selected randomly. We consider the GPS trace of selected vehicles within 1 hour.
For the initial placement at the start of the experiment, we allocate corresponding containers for each vehicle at the same edge data center according to its GPS coordinates. The nearest edge data center first policy is considered for our experiment to generate the live container migration requests. According to the user mobility, one live container migration will be triggered when one vehicle exits the coverage area of its current edge data center. 
There are 9933, 19522, and 37822 migration requests induced by these vehicles' movement, respectively.
During the live container migration, the dirty memory of migrating container will be copied iteratively from the source edge data center to the nearest edge data center through the shortest network path.  For the evaluation sensitivity, the results of each scenario are an average of 10 individual experiments. In this experiment, we consider that the container image as a universal service is already available in all edge data centers or shared by the network storage between EDCs.
CloudSimSDN-NFV  \cite{son2019cloudsimsdn}, an event-driven simulator, is extended with corresponding components to support the live container migration and user mobility in edge computing environments.

\begin{figure*}[th]
	\centering
	\begin{subfigure}{.33\linewidth}
		\centering
		\includegraphics[width=\linewidth]{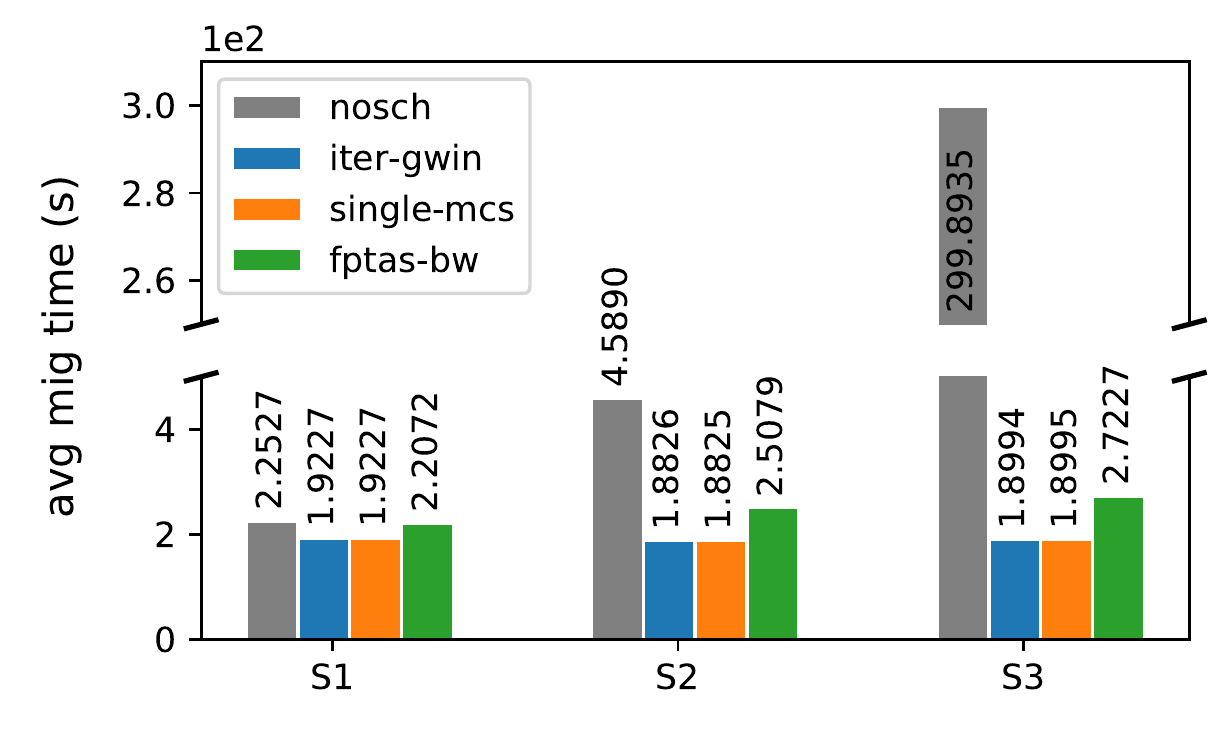}
		\caption{average migration time}
		\label{fig: avg-mig}
	\end{subfigure}%
	\begin{subfigure}{.33\linewidth}
		\centering
		\includegraphics[width=\linewidth]{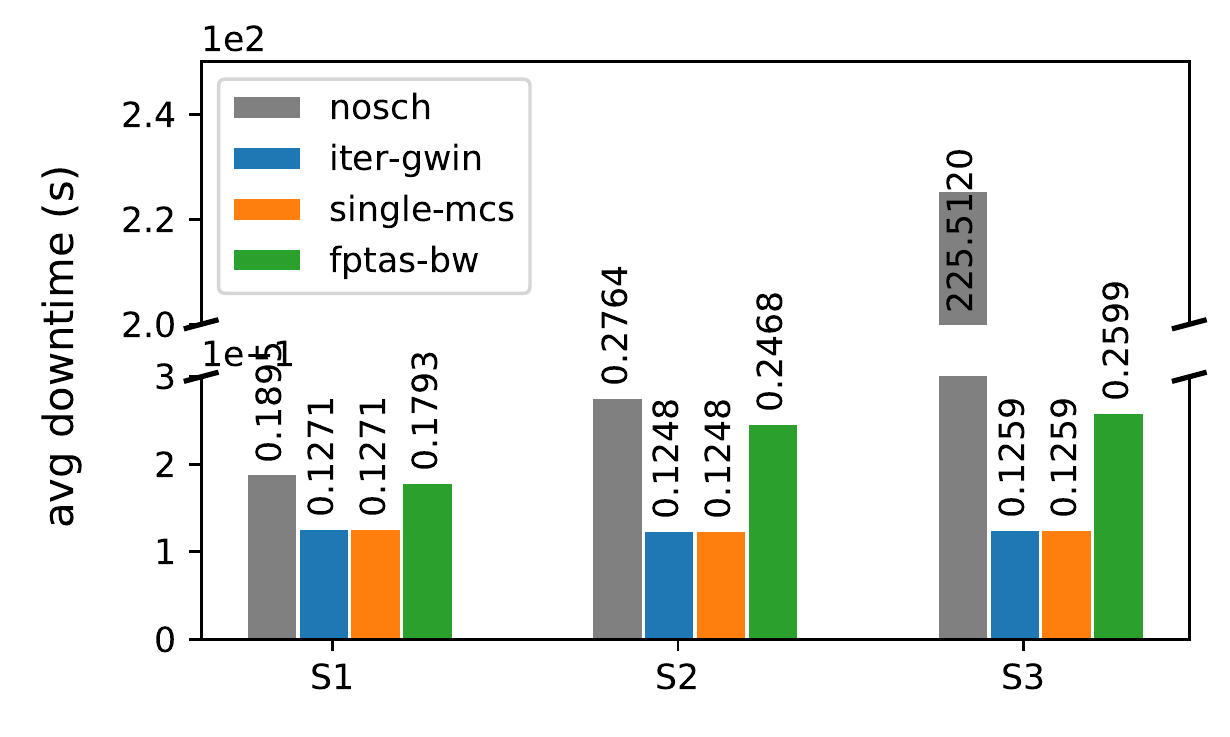}
		\caption{average downtime}
		\label{fig: avg-dt}
	\end{subfigure}
	\begin{subfigure}{.33\linewidth}
		\centering
		\includegraphics[width=\linewidth]{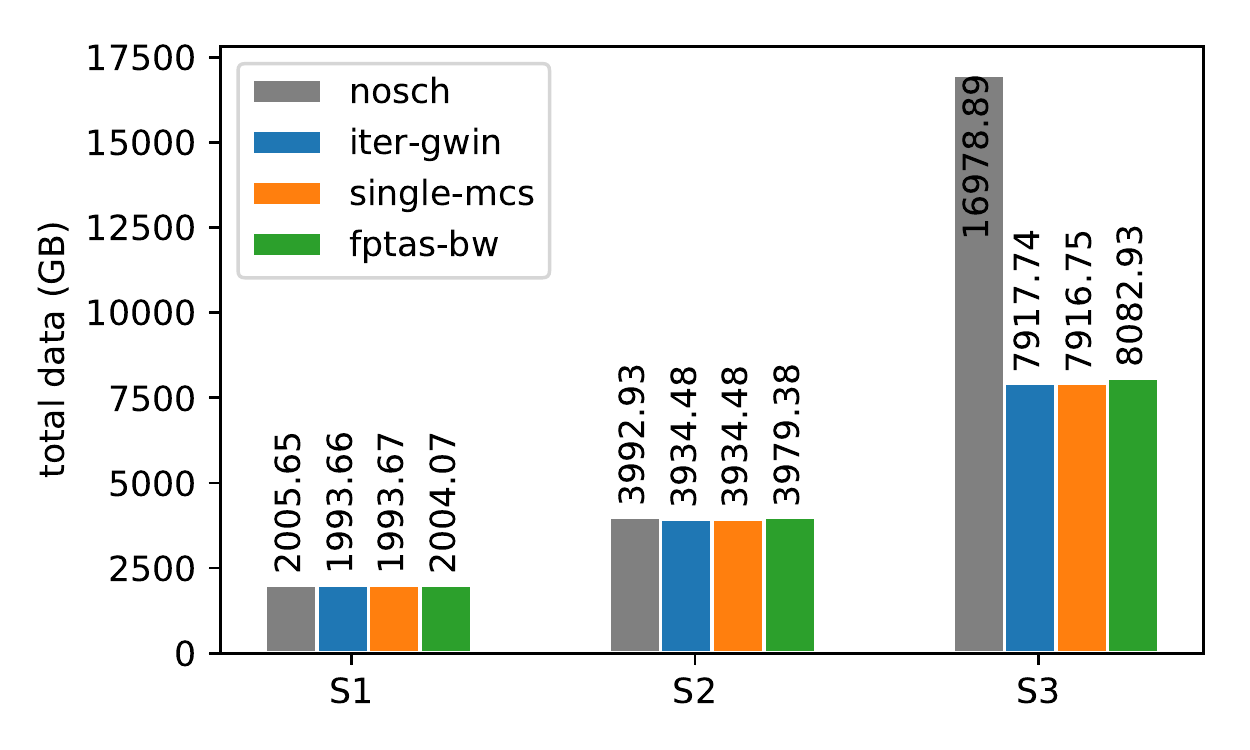}
		\caption{total transferred data}
		\label{fig: total-data}
	\end{subfigure}
	
	\begin{subfigure}{.33\linewidth}
		\centering
		\includegraphics[width=\linewidth]{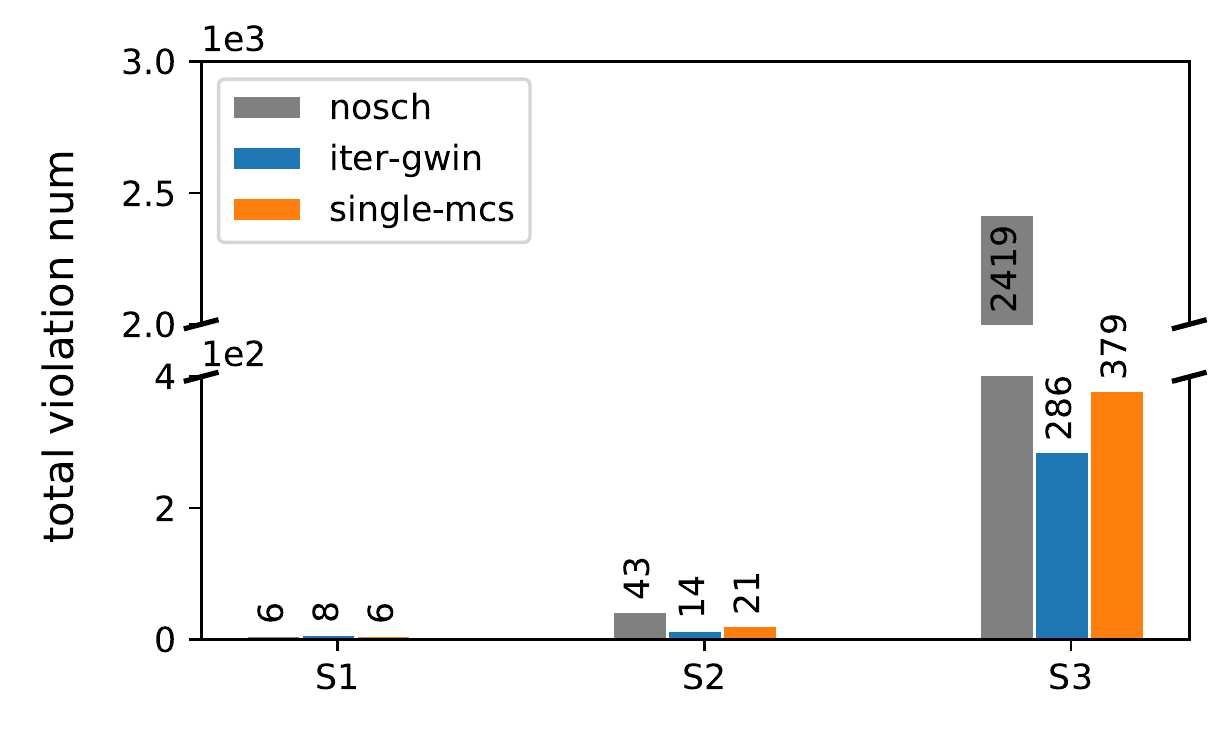}
		\caption{total deadline violations}
		\label{fig: vio-num}
	\end{subfigure}%
	\begin{subfigure}{.33\linewidth}
		\centering
		\includegraphics[width=\linewidth]{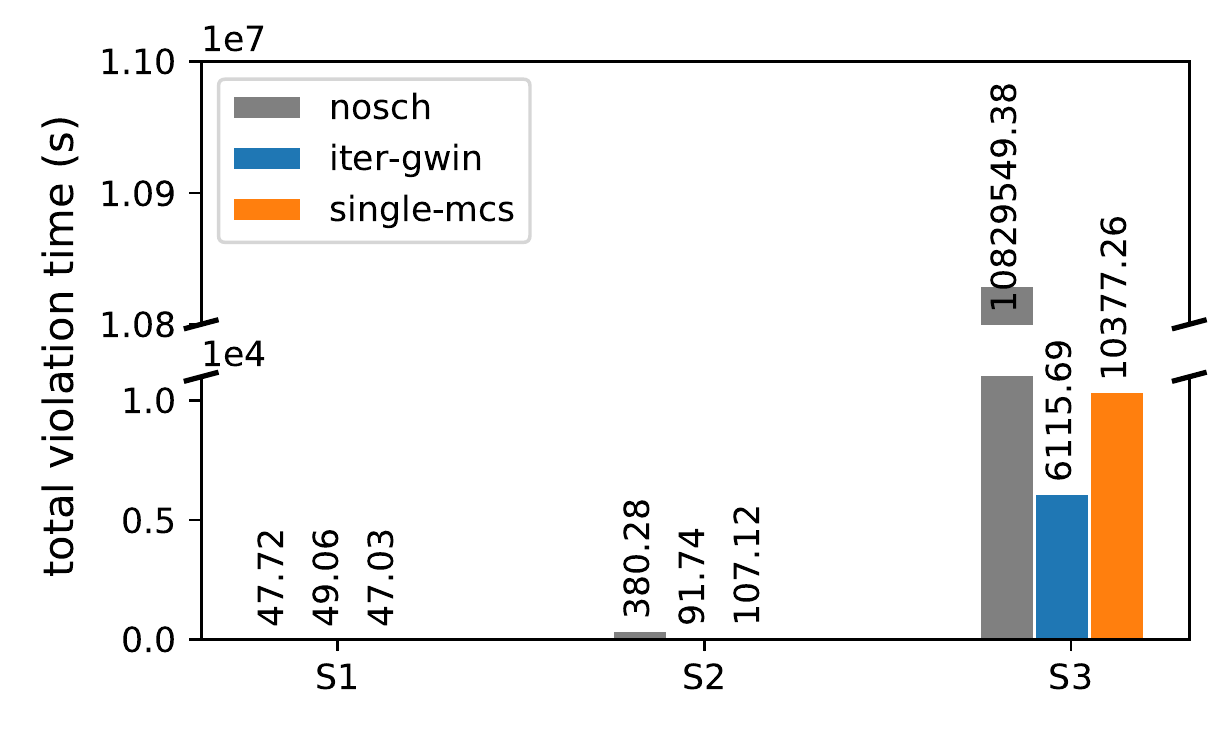}
		\caption{total violation time}
		\label{fig: vio-time}
	\end{subfigure}
	\begin{subfigure}{.33\linewidth}
		\centering
		\includegraphics[width=\linewidth]{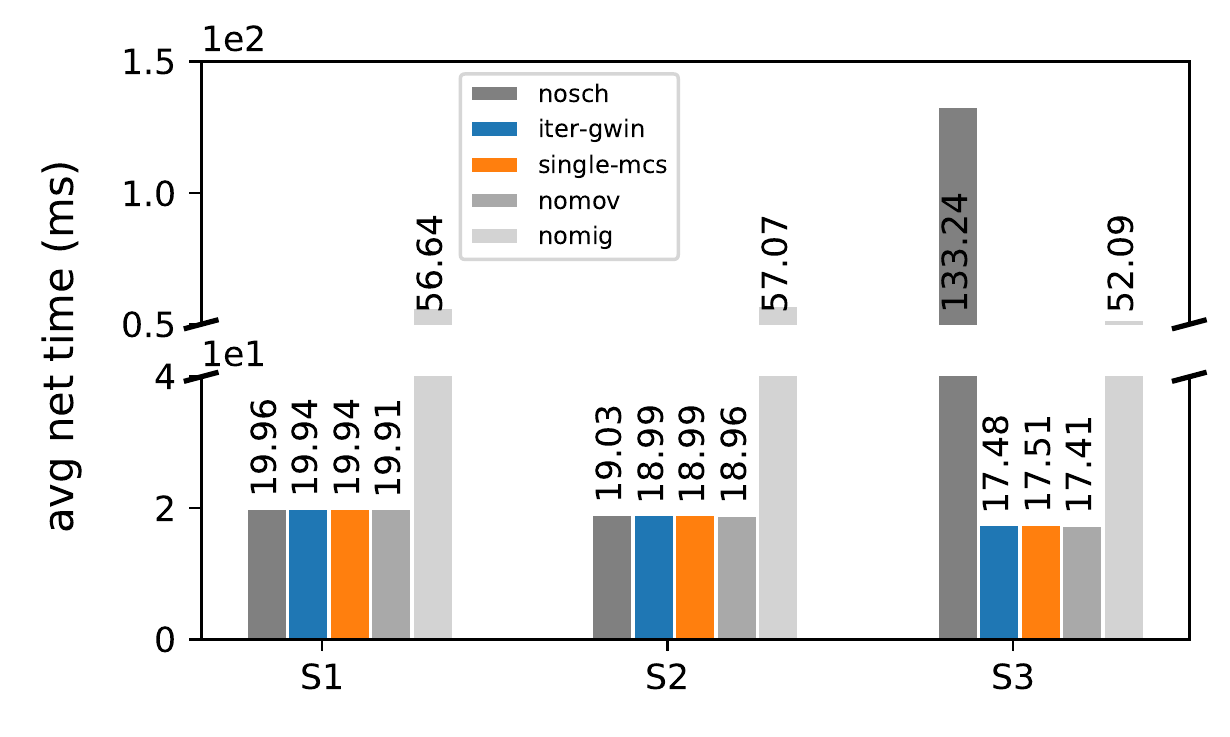}
		\caption{network transmission time}
		\label{fig: avg-net-time}
	\end{subfigure}
	\caption{Migration performance comparison with no scheduler, iter-GWIN, and single-MCs under different scenarios.}
	\label{fig: mig-performance}
\end{figure*}

\subsection{Experimental Results and Analysis}
\begin{table}[t]
	\caption{Total processing time comparison in milliseconds}
	\label{tb: results-proc}
	\centering
	\begin{tabular}{|l|lll|}
		\hline
		algorithm 		&  S1 		  &  	S2			   	   & 		S3 \\
		\hline
		iter-GWIN       & 306.9607  &      762.3356       &	    3997.8309		 \\
		single-MCs      & 332.5682  &      583.3951        &	1544.6291     \\
		FPTAS\cite{wang2017virtualold}&903597.39	  &     1923036.81        &     4677990.57             \\
		\hline
	\end{tabular}
	
\end{table}

In Section \ref{section: evaluation}, we compare the algorithm performance in terms of the size of iterative MISs and processing time. Thus, we only evaluate the two best algorithms in this section. We compare the experimental results between no migration scheduler, iter-GWIN, single-MCs, and the current state-of-the-art migration planning and scheduling in clouds FPTAS \cite{wang2017virtualold}. In FPTAS, to maximize the total bandwidth utilized by migrations, one migration can be started even there is considerably limited bandwidth which is much lower than the dirty page rate per second. The solution can cause devastating migration performance. Thus, we improve FPTAS by adding a bandwidth threshold (FPTAS-BW) that the available bandwidth must larger than the dirty page rate as the migration start condition.
As the vehicle number increases from scenario 1 (S1) to scenario 3 (S3), the density of live migration requests in certain areas increases dramatically. The resource competition or resource dependency among live container migration requests will also increase. As a result, the complexity of the dependency graph may also increase. When the requirements of live container migration requests exceed the resource capacity provided by the edge computing, it is inevitable that some of the deadlines of some migration requests can not be satisfied.

Table \ref{tb: results-proc} shows the total processing time of migration planning and scheduling algorithms within 1 hour in milliseconds. From S1 to S3, the average processing time of single-MCs for each migration planning is 0.1175, 0.1936, and 0.4904 milliseconds. Compared to iter-GWIN, the processing time of single-MCs decreased by $61.36\%$ in scenario S3.
The results are consistent with the algorithm evaluation in Section \ref{section: evaluation}.
Furthermore, compared to FPTAS \cite{wang2017virtualold}, the performance of our solution in terms of processing time has been improved by more than 3000 times. In S3, the processing time of FPTAS is about 78 minutes. As a result, even with any weight modification in the algorithm, the migration deadline in seconds will be missed. Therefore, as the results of FPTAS-BW in deadline violation are off the limit of chart comparison, we only compare it in the migration performance.

From S1 to S3, without migration planning and scheduling, more live migrations compete with each other on the network routing and the available bandwidth. As a result, the average migration time increases dramatically from 2.25 and 4.59 seconds to 299.89 seconds (Fig. \ref{fig: avg-mig}). Particularly, in S3, the allocated bandwidth may either be smaller than the dirty page rate and cause a large downtime for some migrations. Or, it causes a much longer migration time due to a large number of memory-copying iterations. 
As a result, the migrating service suffers a devastating consequence.
Furthermore, for FPTAS-BW, by maximizing the total migration bandwidth rather than the resource competitions, it suffers smaller average bandwidth per migration. Thus, as shown in Fig. \ref{fig: mig-performance} the performance of our purposed solution in terms of average migration time, average downtime, and total transferred data are increased by up to $30.24\%$, $51.56\%$, and $2.06\%$, respectively. 
Meanwhile, for the proposed planning and scheduling algorithms iter-GWIN and single-MCs, the performance of live migration can be guaranteed even with severe resource competitions. Results (Fig. \ref{fig: avg-mig}, \ref{fig: avg-dt}) show that the average migration time and downtime are optimal at 1.9 sand 0.13 seconds as there is no bandwidth sharing between resource-dependent migrations. 
Furthermore, for all the migrations that arrive within the 3600 seconds time interval in S3, iter-GWIN and single-MCs can finish the scheduling of all migrations in 3603.91 and 3601.43 seconds. However, the total migration time of no scheduler is 3603.43 seconds in S2 and 48878.65 seconds in S3.
A shorter average migration means less possibility of QoS degradation and less occupation time on the network resource. A smaller downtime equals fewer disruptions on the migrating services.

Another critical migration performance is the transferred data of the live migration. It is also highly related to network energy efficiency.
In S1 and S2, although average migration time and downtime increase due to less allocated bandwidth, there is no surge in the transferred data for the no migration scheduling situation (Fig. \ref{fig: total-data}). Because of the container's small memory footprint, the shared bandwidth can still satisfy the downtime threshold with relatively small memory-copying iterations.
However, when the bandwidth becomes the bottleneck, a large number of memory-copying iteration needed to meet the downtime threshold. Therefore, the total transferred data in S3 increase by $114.47\%$ compared with the optimal result from single-MCs.

The deadline of a live migration request is highly related to the QoS and SLA requirement of the real-time migrating service.
For iter-GWIN and single-MCs, the ratio of migration violation numbers to the total migration number is $0.071\%$ and $0.107\%$ in S2 and $0.756\%$ and $1.002\%$ in S3 (Fig. \ref{fig: vio-num}). However, the ratio for no migration scheduler is 3.07 times in S2 and 8.46 times compared to the best result from iter-GWIN. The ratio of total violation time to the service time of all containers in one hour is $0.00127\%$ and $0.00148\%$ in S2 and $0.0425\%$ and $0.0720\%$ in S3, respectively (Fig. \ref{fig: vio-time}). In S3, although migration performance in terms of migration time and downtime is optimized by the migration scheduler, the network resource is insufficient to schedule all 37822 migration requests on time with the live migration competitions. It is inevitable to violate the deadline of certain migrations with lower priority to satisfy the deadline for others. As a solution, one needs to increase the network resource by providing duplicate EDCs and additional network routing or available bandwidth in the hot spot to alleviate the deadline violation of real-time migrations.

The end-to-end delay for the migrating edge service is affected by the migration downtime and the duration of deadline violation.
For the network transmission time, we compare the results of no user movement and no migration, no migration requests with user movement, no scheduler, iter-GWIN and single-MCs (Fig. \ref{fig: avg-net-time}).
In the scenario that all vehicles stay at the s and do not move during the experiment time (nomov), the average network transmission time to the service or the end-user is from 17.4 to 19.9 milliseconds from S3 to S1.
Without the live migration requests (nomig), the end-to-end delay can be not guaranteed due to the network delay between the EDC and the end-user. Specifically, the average network transmission time is around 56 milliseconds. The live migration planning and scheduling algorithm (iter-GWIN and single-MCs) can guarantee the average service network transmission time. In S3, without a migration scheduler, the downtime and deadline violation have a considerable impact on the service network delay. The network delay increases by 6.62 times compared to the result of iter-GWIN.

In summary, our proposed algorithms can efficiently plan and schedule large-scale mobility-induced live container migrations in edge computing. Even in the case of a migration request surge, it guarantees the performance of live container migrations and maintains the QoS of migrating services.
It significantly reduces average migration time (up to $99.36\%$), average down time (up to $99.94\%$), total deadline violations (up to $88.18\%$) and violation time (up to $99.94\%$).

\section{Conclusions and Future Work} \label{section: conclusion}
In this paper, we investigated the challenges of live container migration scheduling in edge computing environments including (1) resource competition or dependency among live migrations and (2) real-time migration planning and scheduling.
We modeled the relationship of resource dependency among migrations as an undirected graph.
and the scheduling problem as generating the maximum independent set of the dependency graph iteratively.
We proposed a framework for user-triggered or mobility-induced migration scheduling which is different from the traditional scheduling for live VM migrations in cloud data centers. 
The SDN is introduced to separate the computer network to minimize the impact of migration flows on other edge services.
Based on the dynamic computing resources, network resources and topology provided by the container/VM orchestration engine and SDN controllers, the migration management service can plan and schedule multiple migration requests in a fine-grained manner. 
We proposed two methods for large-scale migration planning and scheduling algorithms based on iterative Maximal Independent Sets. Computational experiments were conducted to evaluate the algorithms' performance. 
Furthermore, the results of experiments based on real-world data indicate that proposed algorithms can efficiently plan and schedule large-scale mobility-induced live container migrations in a complex network environment in a timely manner, while maintaining the QoS of migrating services. It can optimize the live migration performance and minimize the deadline violation in migration scheduling. 
As part of the future work, we intend to investigate base station clustering for EDC placement based on user mobility information to reduce the number of live migrations.

\section*{Acknowledgment}
The authors thank Shashikant Ilager and Mohammad Goudarzi for their valuable comments and suggestions.

\bibliographystyle{IEEEtran}
\bibliographystyle{abbrv}
\bibliography{ref-full}

\end{document}